\documentclass[a4paper,11pt]{article}

\usepackage{xspace}
\usepackage{tikz}

\usepackage{amsmath}
\usepackage{amsfonts}
\usepackage[utf8]{inputenc}
\usepackage[english]{babel}
\usepackage{fullpage}

\usepackage{url,hyperref}

\newcommand\X{\mathbb{X}}
\newcommand\R{\mathbb{R}}
\newcommand\norm[1]{||#1||}

\newcommand\bM[2]{B_{\manifold}(#1,#2)}
\newcommand\Dgm[1]{\mathrm{Dgm}(#1)}

\newcommand\dm{d_{\mu,m}}
\newcommand\dmP{d_{\mu_P,m}}
\newcommand\dn{d_{\nu,m}}

\newcommand\dM{d_{\mu_\manifold,m}}

\newcommand\cV{{\cal {V}}}

\newcommand {\mm}[1] {\ifmmode{#1}\else{\mbox{\(#1\)}}\fi}
\newcommand{\denselist}{\itemsep 0pt\parsep=1pt\partopsep 0pt}
\newcommand{\eps}{{\varepsilon}}
\newcommand{\etal}{et al.\@\xspace}
\newcommand{\mygood}{{accurate\xspace}}

\newcommand{\F}		{{\mathcal{F}}}
\newcommand{\manifold}	{\mathsf{M}}

\newcommand{\extNN}	{\mathrm{NN}}
\newcommand{\funcsample}	{{functional-sample\xspace{}}}
\newcommand{\newfunc}[1]	{{\widehat{{#1}}}}
\newcommand{\discrep}	{{\phi}}
\newcommand{\argmin}{\operatornamewithlimits{argmin}}
\newcommand{\fprecision}		{{\xi}}

\renewcommand{\footnote}{\arabic{footnote}}

\newtheorem{theorem}{Theorem}[section]
\newtheorem{definition}[theorem]{Definition}	
\newtheorem{lemma}[theorem]{Lemma}
\newtheorem{corollary}[theorem]{Corollary}
\newtheorem{claim}[theorem]{Claim}
\newtheorem{obs}[theorem]{Observation}
\newtheorem{proposition}[theorem]{Proposition}

\newenvironment{proof}{{\em Proof:}}{\hfill{\hfill\rule{2mm}{2mm}}}

 \newenvironment{algorithm}[1]{\begin{center}\rule{140mm}{0.3mm}\\\textsc{\textbf{#1}}\end{center}}
{\begin{center}\rule{140mm}{0.3mm}\end{center}}

\title{Topological analysis of scalar fields with outliers}

\author{Micka\"{e}l Buchet\footnote{\texttt{mickael.buchet@m4x.org} - Inria Saclay \^{I}le-de-France}, Fr\'{e}d\'{e}ric Chazal\footnote{\texttt{frederic.chazal@inria.fr} - Inria Saclay \^{I}le-de-France}, Tamal K. Dey\footnote{\texttt{tamaldey@cse.ohio-state.edu} - Department of Computer Science and Engineering, The Ohio State University},\\ Fengtao Fan\footnote{\texttt{fan.171@osu.edu} - Department of Computer Science and Engineering, The Ohio State University}, Steve Y. Oudot\footnote{\texttt{steve.oudot@inria.fr} - Inria Saclay \^{I}le-de-France}, Yusu Wang\footnote{\texttt{yusu@cse.ohio-state.edu} - Department of Computer Science and Engineering, The Ohio State University}}

\begin{document}

 \setcounter{page}{0}
 \thispagestyle{empty}
 \maketitle
\begin{abstract}
Given a real-valued function $f$ defined over a manifold $\manifold$
embedded in $\R^d$, we are interested in recovering structural
information about $f$ from the sole information of its values on a
finite sample $P$. Existing methods provide
approximation to the persistence diagram of $f$ when geometric noise and functional noise are
bounded. However, they fail in the
presence of aberrant values, also called outliers, both in theory and practice.

We propose a new algorithm that deals with outliers. We handle aberrant
functional values with a method inspired from the k-nearest neighbors
regression and the local median filtering, while the geometric outliers
are handled using the distance to a measure. Combined with topological
results on nested filtrations, our algorithm performs robust topological
analysis of scalar fields in a wider range of noise models than handled
by current methods. We provide theoretical guarantees and experimental
results on the quality of our approximation of the sampled scalar field.
\end{abstract}

  
 \clearpage

\section{Introduction}
Consider a network of sensors measuring a quantity such as the temperature, the humidity, or the elevation.
These sensors also compute their positions and communicate these data to others.
However, they are not perfect and can make mistakes such as 
providing some aberrant values.
Can we still recover topological structure from the measured quantity?

This is an instance of a scalar field analysis problem.
Given a manifold $\manifold$ embedded in $\R^d$ and a scalar field $f:\manifold\rightarrow\R$, we want to extract topological information about $f$, knowing only its values on a finite set of points $P$.
The critical points of a function, 
that is, peaks (local maxima), pits (local minima), and passes (saddle points)
constitute important topological features of the function. 
In addition, the prominence of these features also
contains valuable information, which the geographers use to 
distinguish between a summit and a local maximum in its shadow.
Such information can be captured by the so-called \emph{topological persistence}, 
which studies the \emph{sub-level sets} $f^{-1}((-\infty,\alpha])$ of a function $f$ and the way their topology evolves as parameter $\alpha$ increases.
In the case of geography, we can use the negated
elevation as a function to study the topography.
Peaks will appear depending on their altitude and will merge into other topological features at saddle points.
This provides a \emph{persistence diagram} describing the lifespan of features 
where the peaks with more prominence have longer lifespans.

When the domain $\manifold$ of the function $f$ is triangulated, one classical way of computing this diagram is to linearly interpolate the function $f$ on each simplex and then apply the standard persistence algorithm to this piecewise-linear function \cite{cphGZ}. 
For cases where we only have pairwise distances between input points, one can build a family of simplicial complexes and infer the persistent homology of the input function $f$ from them \cite{sfapcdCGOS} (this construction will be detailed in Section \ref{sec:scalarpreliminary}).

Both of these approaches can provably approximate persistent homology when the input points admit a bounded noise, i.e., when 
the Hausdorff distance between $P$ and $\manifold$ is bounded and the $L_\infty$-error on the observed value of $f$ is also bounded.
What happens if the noise is unbounded?
A faulty sensor can provide completely wrong information or a bad position. 
Previous methods no longer work in this setting.
Moreover, a sensor with a good functional value but a bad position 
can become an outlier in function value 
at its measured position (see Section \ref{ssec:functional} for an example). 
In this paper, we study the problem of analyzing scalar fields in the
presence of unbounded noise both in the geometry and in the functional values.
To the best of our knowledge, there is no other method to handle such combined unbounded geometric and functional noise with theoretical guarantees. 

\paragraph*{Contributions.}
We consider a general sampling condition. 
Intuitively, a sample $(P,\tilde f)$ of a function $f: \manifold \to \R$ respects our condition if: 
(i) the domain $\manifold$ is sampled densely and there is no cluster of noisy samples outside $\manifold$ (roughly speaking, no area outside $\manifold$ has a higher sampling density than on $\manifold$), and (ii) for any point of $P$, at least half of its $k$ nearest neighbors have a functional value 
with an error less than a threshold $s$.
This condition allows functional outliers that may have a value 
arbitrarily far away from the true one.
It encompasses the previous bounded sampling conditions as well 
as other sampling conditions such as bounded Wasserstein distance for geometry, 
or generative models like an additive Gaussian noise. 
Connection to some of these classical sampling conditions can be found in Appendices~\ref{sec:relNoiseFunc} and \ref{sec:relNoiseGeom}.

We show how to approximate the persistence diagram of $f$ knowing only its
observed value $\tilde f$ on the set $P$.
We achieve this goal through three main steps:
\begin{enumerate}\denselist
\item Using the observations $\tilde f$, we provide a new 
estimator $\hat f$ to approximate $f$. 
This estimator is inspired by the $k$-nearest neighbours regression technique but differs from it in an essential way.
\item We filter geometric outliers using a distance to a measure function.
\item We combine both techniques in a unified framework to estimate the persistence diagram of $f$.
\end{enumerate}
The two sources of noise, geometric and functional, are not independent.
The interdependency is first identified by assuming
appropriate sampling conditions, and then untangled by separate
steps in our algorithm.

\paragraph*{Related work.}
A framework for scalar field topology inference with 
theoretical guarantees has been previously proposed in 
\cite{sfapcdCGOS}.
However, it is limited to a bounded noise assumption, which we aim to relax. 

For handling the functional noise only, the traditional non-parametric 
regression mostly uses kernel-based or $k$-NN estimators.
The $k$-NN methods are more versatile~\cite{dftnrG}. Nevertheless,
the kernel-based estimators are preferred when there is structure in the data.
However, the functional outliers destroy the structure on which kernel-based estimators rely. 
These functional outliers can arise as a result of geometric outliers (see Section \ref{ssec:functional}). Thus, in a way, it is essential to be able to handle functional outliers when the input has geometric noise. 
Functional outliers can also introduce a bias that hampers the robustness of a $k$-NN regression.
For example, if all outliers' values are greater than the actual value, a $k$-NN regression will shift towards a larger value.
Our approach leverages the $k$-NN regression idea while trying to avoid the 
sensitivity to this bias.
 
Various methods for geometric denoising have also
been proposed in the literature. 
If the generative model for noise is known a priori, 
one can use de-convolution to remove noise.
Some methods have been specifically adapted to use topological information for such denoising~\cite{tdstsCK}.
In our case where the generative model is unknown, we use a filtering by the value of the distance to a measure, which has been successfully 
applied to infer the topology of a domain under unbounded noise \cite{gipmCCM}.

\section{Preliminaries for Scalar Field Analysis}
\label{sec:scalarpreliminary}

In \cite{sfapcdCGOS}, Chazal \etal{} presented an algorithm to analyze the scalar field topology using persistent homology which can handle bounded Hausdorff  noise both in geometry and in observed function values. 
Our approach follows the same high level framework. 
Hence in this section, we introduce necessary preliminaries 
along with some of the results from \cite{sfapcdCGOS}. 
 

\paragraph*{Riemannian manifold and its sampling. }
Consider a compact Riemannian manifold $\manifold$. Let  $d_\manifold$ denote the geodesic metric on $\manifold$. 
Consider the open Riemannian ball $\bM{x}{r} := \{ y \in \manifold \mid d_\manifold(x, y) < r\}$ centered at $x \in \manifold$.
$\bM{x}{r}$ is \emph{strongly convex} if for any pair $(y,y')$ in the closure of $\bM{x}{r}$, there exists a unique minimizing
geodesic between $y$ and $y'$ whose interior is contained in $\bM{x}{r}$.
Given any $x \in \manifold$, let $\varrho(x)$ denote the supremum of the value of $r$ such that $\bM{x}{r}$ is strongly convex.
As $\manifold$ is compact, the infimum of all $\varrho(x)$ is positive and we denote it by $\varrho(\manifold)$, which is called the \emph{strong convexity radius} of $\manifold$.

A point set $P \subseteq \manifold$ is a \emph{geodesic $\eps$-sample} of $\manifold$ if for every point $x$ of $\manifold$, the distance from $x$ to $P$ is less than $\eps$ in the metric $d_\manifold$. 
Given a $c$-Lipschitz scalar function $f:\manifold\rightarrow\R$, we aim to study the persistent homology of $f$. 
However, the scalar field $f: \manifold \to \R$ is only approximated by a discrete set of sample points $P$ and a function $\tilde{f}:P\rightarrow\R$.
The goal of this paper is to retrieve the topological structure of $f$ 
from $\tilde{f}$ when some forms of  noise are present both in the positions of $P$ and in the function values of $\tilde{f}$. 

\paragraph*{Persistent homology.}
As in \cite{sfapcdCGOS}, we infer the persistent homology of $f$ using well-chosen \emph{persistence modules}. 
A \emph{filtration} $\{F_\alpha\}_{\alpha\in\R}$ is a family of sets $F_\alpha$ totally ordered by inclusions $F_\alpha \subseteq F_\beta$. 
Following \cite{CCG09}, a persistence module is a family of vector spaces $\{\Phi_\alpha\}_{\alpha\in\R}$ with a family of homomorphisms $\phi_\alpha^\beta:\Phi_\alpha\rightarrow\Phi_\beta$ such that for all $\alpha\leq\beta\leq\gamma$, $\phi_\alpha^\gamma=\phi_\beta^\gamma\circ\phi_\alpha^\beta$. 
Given a filtration $\F = \{F_\alpha\}_{\alpha\in\R}$ and $\alpha\leq\beta$, the canonical inclusion $F_\alpha\hookrightarrow F_\beta$ induces a homomorphism at the homology level $H_*(F_\alpha)\rightarrow H_*(F_\beta)$.
These homomorphisms and the homology groups of $F_\alpha$ form the so-called \emph{persistence module} of $\F$.  

The persistence module of the filtration $\F = \{F_\alpha\}_{\alpha\in\R}$ is said to be 
\emph{q-tame} when all the homomorphisms $H_*(F_\alpha)\rightarrow H_*(F_\beta)$ have finite rank \cite{CDGO13}.
Its algebraic structure can then be described by the \emph{persistence diagram} $\Dgm \F$, which is a multiset of points in $\R^2$ describing the lifespan of the homological features in the filtration $\F$.
For technical reasons, $\Dgm \F$ also contains every point of the diagonal $y=x$ with countably infinite multiplicity. See \cite{EH09} for a more formal discussion of the persistence diagrams. 

Persistence diagrams can be compared using the \emph{bottleneck distance} $d_B$ \cite{CEH07}.
Given two multisets with the same cardinality, possibly infinite, $D$ and $E$ in $\R^2$, we consider the set $\cal{B}$ of all bijections between $D$ and $E$.
The bottleneck distance (under $L_\infty$-norm) is then defined as:
\begin{eqnarray}
d_B(D,E)=\inf_{b\in {\cal B}}\sup_{x\in D}\norm{x-b(x)}_\infty.\label{eqn:bottleneck-def}
\end{eqnarray}

Two filtrations $\{U_\alpha\}$ and $\{V_\alpha\}$ are said to be \emph{$\eps$-interleaved} if, for any $\alpha$, we have $U_\alpha\subset V_{\alpha+\eps}\subset U_{\alpha+2\eps}.$ 
Recent work in \cite{CCG09,CDGO13} shows that two interleaved filtrations induce close persistence diagrams in the bottleneck distance. 
\begin{theorem}\label{tstability}
Let $U$ and $V$ be two $q$-tame and $\eps$-interleaved filtrations. Then the persistence diagrams of these filtrations verify $d_B(\Dgm{U},\Dgm{V})\leq\eps.$
\end{theorem}


\paragraph*{Nested filtrations.}
The scalar field topology of $f: \manifold \to \R$ is studied via the topological structure of the sub-level sets filtration of $f$. 
More precisely, the sub-level sets of $f$ are defined as $F_\alpha=f^{-1}((-\infty,\alpha])$ for any $\alpha \in \R$.
The collection of sub-level sets forms a filtration $\F = \{F_\alpha\}_{\alpha\in\R}$ connected by natural inclusions $F_\alpha \subseteq F_\beta$ for any $\alpha\leq\beta$. 
Our goal is to approximate the persistence diagram $\Dgm \F$ from the observed scalar field $\tilde{f}: P \rightarrow \R$. 
We now describe the results of \cite{sfapcdCGOS} for approximating $\Dgm\F$ when $P$ is a geodesic $\eps$-sample of $\manifold$. These results will later be useful for our approach.

To simulate the sub-level sets filtration $\{F_\alpha\}$ of $f$, we introduce $\ P_\alpha=\tilde{f}^{-1}((-\infty,\alpha])\subseteq P$ for any $\alpha \in \R$. 
The points in $P_\alpha$ intuitively sample the sub-level set $F_\alpha$. 
To estimate the topology of $F_\alpha$ from these discrete samples $P_\alpha$, 
we consider the \emph{$\delta$-offset} $P^\delta$ of the point set $P$, i.e., we grow geodesic balls of radius $\delta$ around the points of $P$.
This gives us a union of balls that serves as a 
proxy for $f^{-1}((-\infty,\alpha])$. The nerve of this
collection of balls, also known as the \emph{\v Cech complex}, $C_\delta(P)$,
has many interesting properties but is difficult to compute in high dimensions.
We consider an alternate complex called the \emph{Vietoris-Rips complex} $R_\delta(P)$ that is easier to compute. It is defined as the maximal simplicial complex with the same 1-skeleton as the \v Cech complex. The \v Cech  and Rips complexes are related in any metric space: $\forall \delta>0,\ C_\delta(P)\subset R_\delta(P)\subset C_{2\delta}(P).$

Even though a single Vietoris-Rips complex may not capture the homology of the manifold $\manifold$, 
a pair of nested complexes can recover it 
using the inclusions $R_\delta(P_\alpha)\hookrightarrow R_{2\delta}(P_\alpha)$~\cite{tpbresCO}.
Specifically, for a fixed $\delta > 0$, consider the following commutative diagram induced by inclusions, for $\alpha \le \beta$: 
\begin{center}
\begin{tikzpicture}[scale=.7]
\draw (-1,1.5) node {$H_*(R_{\delta}(P_\beta))$};
\draw (-6,1.5) node {$H_*(R_{\delta}(P_\alpha))$};
\draw (-6,3) node {$H_*(R_{2\delta}(P_\alpha))$};
\draw (-1,3) node {$H_*(R_{2\delta}(P_\beta))$};
\draw[->] (-4.4,1.5) -- (-2.6,1.5);
\draw[->] (-4.4,3) -- node[above] {$\phi_\alpha^\beta$} (-2.6,3);
\draw[->] (-6,1.8) -- node[left] {$i_\alpha$} (-6,2.5);
\draw[->] (-1,1.8) -- node[right] {$i_\beta$} (-1,2.5);
\end{tikzpicture}
\end{center}
As the diagram commutes for all $\alpha\leq\beta$, $\{ Im(i_\alpha), \phi_\alpha^\beta|_{Im(i_\alpha)}\}$ defines a persistence module.
We call it the persistent homology module of the filtration of nested pairs $\{R_\delta(P_\alpha)\hookrightarrow R_{2\delta}(P_\alpha)\}_{\alpha\in\R}$.
This construction can also be done for any filtration of nested pairs.
Using this construction,  one of the main results of \cite{sfapcdCGOS} is: 

\begin{theorem}[Theorems 2 and 6 of \cite{sfapcdCGOS}]
\label{tPrevious}
Let $\manifold$ be a compact Riemannian manifold and let $f:\manifold\rightarrow\R$ be a $c$-Lipschitz function.
Let $P$ be a geodesic $\eps$-sample of $\manifold$.
If $\eps<\frac{1}{4}\varrho(\manifold)$, then for any $\delta\in\left[2\eps,\frac{1}{2}\varrho(\manifold)\right)$, the persistent homology modules of $f$ and of the filtration of nested pairs $\{R_\delta(P_\alpha)\hookrightarrow R_{2\delta}(P_\alpha)\}$ are $2c\delta$-interleaved.
Therefore, the bottleneck distance between their persistence diagrams is at most $2c\delta$.

Furthermore, the $k$-dimensional persistence diagram for the filtrations of
nested pairs $\{R_\delta(P_\alpha)\hookrightarrow R_{2\delta}(P_\alpha)\}$ can be computed in $O(|P|kN+N\log N+N^3)$ time, where $N$ is the number of simplices of $\{R_{2\delta}(P_{\infty})\}$, and $|P|$ denotes the cardinality of the sample set $P$. 
\end{theorem}
It has been observed that, in practice, the persistence algorithm often has a running time linear in the number of simplices, which reduces the above complexity to $O(|P|+N\log N)$ in a practical setting.

%
%

We say that $\tilde{f}$ has a precision of $\xi$ over $P$ if $|\tilde{f}(p) - f(p)| \le \xi$ for any $p\in P$. We then have the following result for the case when we only have this functional noise: 

\begin{theorem}[Theorem 3 of \cite{sfapcdCGOS}]
Let $\manifold$ be a compact Riemannian manifold and let $f:\manifold \rightarrow\R$ be a $c$-Lipschitz function.
Let $P$ be a geodesic $\eps$-sample of $\manifold$ such that the values of $f$ on $P$ are known with precision $\xi$.
If $\eps<\frac{1}{4}\varrho(\manifold)$, then for any $\delta\in\left[2\eps,\frac{1}{2}\varrho(\manifold)\right)$, the persistent homology modules of $f$ and of the filtration of nested pairs $\{R_\delta(P_\alpha)\hookrightarrow R_{2\delta}(P_\alpha)\}$ are $(2c\delta+\xi)$-interleaved.
Therefore, the bottleneck distance between their persistence diagrams is at most $2c\delta+\xi$.
\label{thm:funcnoiseonlyfromCGOS}
\end{theorem}

Geometric noise was considered in the form of bounded noise in the estimate of the geodesic distances between points in $P$.
It translated into a relation between the measured pairwise distances and the real ones.
With only geometric noise, one has the following stability result. 
It was stated in this form in the conference version of the paper.

\begin{theorem}[Theorem 4 of \cite{sfapcdCGOS}]\label{tGeomNoise}
Let $\manifold$, $f$ be defined as previously and $P$ be 
an $\eps$-sample of $\manifold$ in its Riemannian metric. 
Assume that, for a parameter $\delta>0$,
the Rips complexes $R_\delta(\cdot)$ are defined with 
respect to a metric ${\tilde d}(\cdot,\cdot)$ which satisfies $\forall x,y\in P,\ \frac{d_\manifold(x,y)}{\lambda}\leq {\tilde d}(x,y)\leq\nu+\mu\frac{d_\manifold(x,y)}{\lambda}$, where $\lambda\geq1$ is a scaling factor, $\mu\geq1$ is a relative error and $\nu\geq0$ an additive error.
Then, for any $\delta\geq\nu+2\mu\frac{\eps}{\lambda}$ and any $\delta'\in[\nu+2\mu\delta,\ \frac{1}{\lambda}\varrho(\manifold)]$, the persistent homology modules of $f$ and of the filtration of nested pairs $\{R_\delta(P_\alpha)\hookrightarrow R_{\delta'}(P_\alpha)\}$ are $c\lambda\delta'$-interleaved.
Therefore, the bottleneck distance between their persistence diagrams is at most $c\lambda\delta'$.
\end{theorem}


\section{Functional Noise}
\label{sec:funconly}

In this section, we focus on the case where we have only functional noise in the observed function $\tilde{f}$. 
Suppose we have a scalar function $f$ defined on a Riemannian manifold $\manifold$ embedded in $\R^d$.
Note that the results of section~\ref{sec:funconly} hold if $\R^d$ is replaced by a metric space $\X$. 
We are given a geodesic $\eps$-sample $P \subset \manifold$, and a noisy observed function $\tilde{f}: P \rightarrow \R$. 
Our goal is to approximate the persistence diagram $\Dgm{\F}$ of the sub-level set filtration $\F = \{ F_\alpha = f^{-1}((-\infty, \alpha]) \}_\alpha$ from $\tilde{f}$. 
We assume  that $f$ is $c$-Lipschitz with respect to the intrinsic metric of the Riemannian manifold $\manifold$. 
Note that this does not imply a Lipschitz condition on $\tilde f$.


\subsection{Functional sampling condition}
\label{ssec:functional}

Previous work on functional noise focused on bounded noise (e.g, \cite{sfapcdCGOS}) or noise with zero-mean (e.g, \cite{kralidK}). 
However, there are many practical scenarios where the observed function $\tilde{f}$ may contain these previously considered types of noise combined with \emph{aberrant function values} in $\tilde{f}$. 
Hence, we propose below a more general sampling condition that allows such combinations. 

\paragraph*{Motivating examples.}
First, we provide some motivating examples for the need of handling \emph{aberrant} function values in $\tilde{f}$, where $\tilde{f}(p)$ at some sample point $p$ can be totally unrelated to the true value $f(p)$.  
Consider a sensor network, where each node returns some measures. Such measurements can be imprecise, and in addition to that, a sensor may experience failure and return a completely wrong measure that has no relation with the true value of $f$. 
Similarly, an image could be corrupted with impulse noise where there are random pixels with aberrant function values, such as random white or black dots.

More interestingly, outliers in function values can naturally appear as a result of (extrinsic) geometric noise present in the discrete samples. 
For example, imagine that we have a process that can measure the function value $f: \manifold \to \R$ with \emph{no error}. However, the geometric location $\tilde{p}$ of a point $p \in \manifold$ can be wrong. In particular, $\tilde{p}$ can be close to other parts of the manifold, thereby although $\tilde{p}$ has the correct function value $f(p)$, it becomes a functional outlier among its neighbors (due to the wrong location of $\tilde{p}$). 
See Figure~\ref{fBonePoints} for an illustration.
The function defined on this bone-like curve is the geodesic distance
to a base point. 
The two sides of the narrow neck have very different function values. 
Now, suppose that the points are sampled uniformly on $\manifold$ and their position is then perturbed by an additive Gaussian noise. 
Then, points from one side of this neck can be sent closer to the other side, causing aberrant values in the observed function. 

In fact, even if we assume that we have a ``magic filter'' that can project each sample back to the closest point on the underlying manifold $\manifold$, the result is a new set of samples where all points are on the manifold and thus can be seen as having {\bf no} geometric noise; however, this point set now contains functional noise which is actually caused by the original geometric noise.
Note that such a magic filter is the goal of many geometric denoising methods.
A perfect algorithm in this sense cannot remove or may even cause more aberrant functional noise.
This motivates the need for handling functional outliers (in addition to traditional functional noise) as well as processing noise
that combines geometric and functional noise together and that 
does not necessarily have zero-mean. 

Another case where our approach is useful concerns with missing data.
Assuming that some of the functional values are missing, we can replace them by anything and act as if they were outliers.
Without modifying the algorithm, we obtain a way to handle the local loss of information.

\begin{figure}[htbp]
\begin{center}
\fbox{
\begin{tabular}{ccc}
\includegraphics[height=2.2cm]{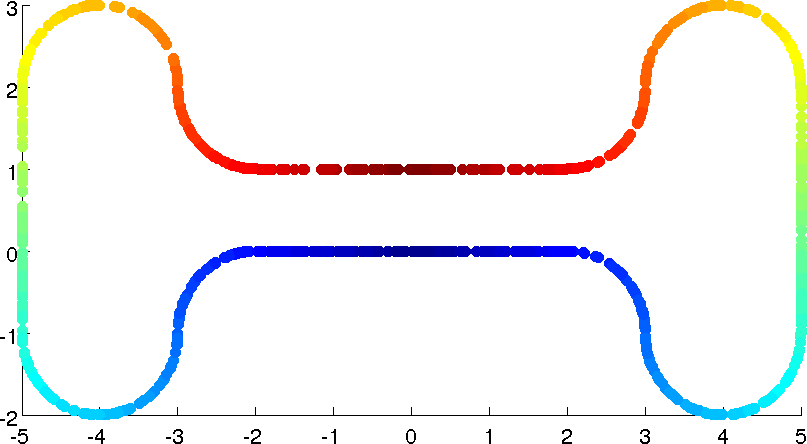} &
\includegraphics[height=2.2cm]{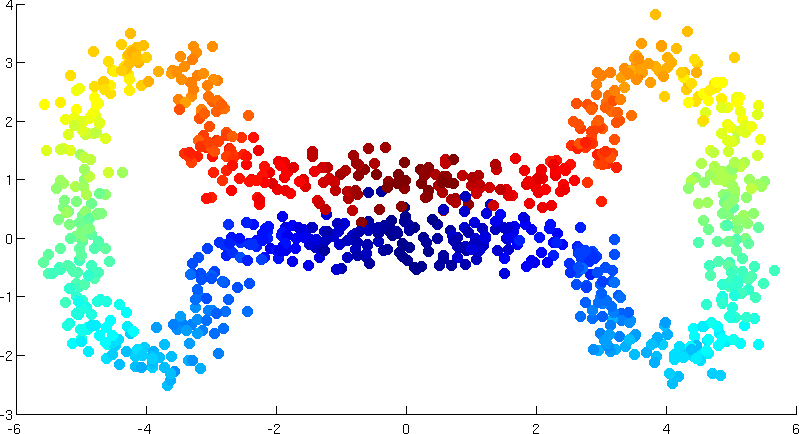} &
\includegraphics[height=2.2cm]{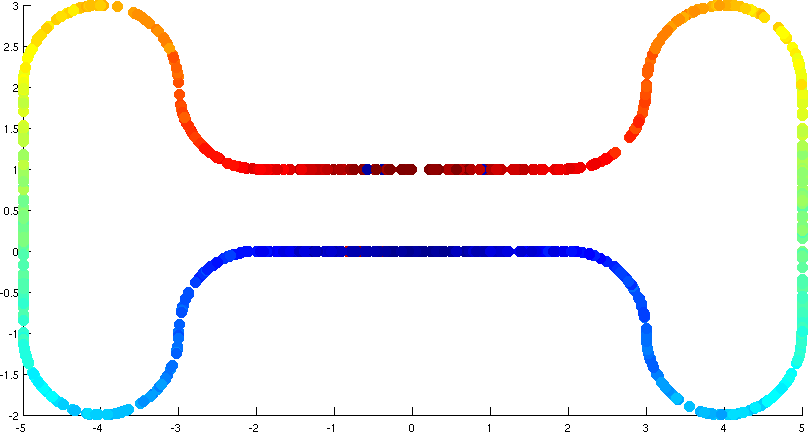} \\
 {Bone without noise} & {Bone with gaussian noise} & {Bone after magical filter}
\end{tabular}
}
\end{center}
\vspace*{-0.2in}
\caption{Bone example after applying Gaussian perturbation and magical filter}
\label{fBonePoints}
\end{figure}
\paragraph*{Functional sampling condition.}
To allow both aberrant and more traditional functional noise, we introduce the following sampling condition. 
Let $P \subset \manifold$ be a geodesic $\eps$-sample of the underlying manifold $\manifold$. 
Intuitively, our sampling condition requires that for every point $p \in P$, locally there is a sufficient number of sample points with reasonably good function values. 
Specifically, we fix two parameters $k$ and $k'$ with the condition that $k\geq k'>\frac{1}{2}k$. Let $\extNN_P^k(p)$ denote the set of the $k$-nearest neighbors of $p$ in $P$ in the \emph{extrinsic metric}. 
We say that a discrete scalar field $\tilde{f}: P \to \R$ is \emph{a $(k, k', \Delta)$-\funcsample{}} of $f: \manifold\to \R$ if the following holds:
\begin{equation}\label{eqn:funcnoisemodel}
\forall p\in P,\ \left | \left \{ q \in \extNN_P^k(p) ~\big |~ |\tilde{f}(q) - f(p)| \le \Delta \right \} \right | \geq k'
\end{equation}
Intuitively, this sampling condition allows up to $k - k'$ samples around a point $p$ to be outliers (whose function values deviates from $f(p)$ by at least $\Delta$). 
In Appendix~\ref{sec:relNoiseFunc}, we consider two standard functional sampling conditions used in the statistical learning community and look at what they correspond to in our setting.

\subsection{Functional Denoising}
Given a scalar field $\tilde{f}: P \to \R$ which is a $(k, k', \Delta)$-\funcsample{} of $f: \manifold \to \R$, we now aim to compute a denoised function $\newfunc{f}: P \to \R$ from the observed function $\tilde{f}$, and we will later use $\newfunc{f}$ to infer the topology of $f: \manifold \to \R$.  Below we describe two ways to denoise the noisy observation $\tilde{f}$: one of which is well-known, and the other one is new. As we will see later, these two treatments lead to similar theoretical guarantees in terms of topology inference. However, they have different characteristics in practice, which are discussed in Appendix~\ref{sec:illustration}.

\paragraph*{$k$-median denoising.}
In the k-median treatment, we simply perform the following: given any point $p \in P$, we set $\newfunc{f} (p)$ to be the median value of the set of $\tilde{f}$ values for the $k$-nearest neighbors $\extNN_P^k(p) \subseteq P$ of $p$. We call $\newfunc{f}$ the \emph{k-median denoising of $\tilde{f}$}. The following observation is straightforward: 
\begin{obs}
If $\tilde{f}: P \to \R$ is a $(k, k', \Delta)$-\funcsample{} of $f: \manifold \to \R$ with $k' \ge k / 2$, then we have $| \newfunc{f}(p) - f(p) | \le \Delta$ for any $p \in P$, where $\newfunc{f}$ is the k-median denoising of $\tilde{f}$. 
\label{obs:kmedian-bound}
\end{obs}

\paragraph*{Disparity-based denoising.}
In the k-median treatment, we choose a single value from the $k$-nearest neighbors of a sample point $p$ and set it to be the denoised value $\newfunc{f}(p)$. This value, while within $\Delta$ distance to the true value $f(p)$ for $k' \ge k/2$, tends to 
have greater variability among neighboring sample points. 
Intuitively, taking the average (such as $k$-means) makes the function $\newfunc{f}(p)$ smoother, but it is sensitive to outliers. We combine these ideas together, and use the following concept of disparity to help us identify a subset of points from the k-nearest neighbors of a sample point $p$ to estimate $\newfunc{f}(p)$. 

Given a set $Y = \{ x_1, \ldots, x_l\}$ of $l$ sample points from $P$, we define its disparity w.r.t. $\tilde{f}$ as: 
$$\discrep(Y) = \frac{1}{l} \sum_{i=1}^l (\tilde{f}(x_i) - \mu(Y))^2, ~~~~\text{where}~~\mu(Y) = \frac{1}{l}\sum_{i=1}^l \tilde{f}(x_i). $$ 
$\mu(Y)$ and $\discrep(Y)$ are respectively the average and the variance of the observed function values for points from $Y$. 
Intuitively, $\discrep(Y)$ measures how tight the function values $(\tilde{f}(x_i))$ are clustered. 
Now, given a point $p \in P$, we define
$$\newfunc{Y}_p = \argmin_{Y \subseteq \extNN_P^k(p), |Y| = k' } ~\discrep(Y), ~~~\text{and}~~ \newfunc{z}_p = \mu(\newfunc{Y}_p). $$
That is, $\newfunc{Y}_p$ is the subset of $k'$ points from the $k$-nearest neighbors of $p$ that has the smallest disparity and $\hat z_p$ is its mass center. 
It turns out that $\newfunc{Y}_p$ and $\newfunc{z}_p$ can be computed by the following {\em sliding-window} procedure: 
(i) Sort $\extNN_P^k(p) =  \{ x_1, \ldots, x_k \}$ according to $\tilde{f}(x_i)$. (ii) For every $k'$ consecutive points $Y_i = \{ x_i,\ldots, x_{i+k'-1} \}$ with $i \in [1, k-k'+1]$, compute its disparity $\discrep(Y_i)$.
 (iii) Set $\newfunc{Y}_p = \argmin_{Y_i, i\in [1, k-k']} \discrep(Y_i)$, and return $\mu(\newfunc{Y}_p)$ as $\newfunc{z}_p$. 

In the \emph{disparity-based denoising} approach, we simply set $\newfunc{f}(p) := \newfunc{z}_p$ as computed above. 
The approximation guarantee of $\hat f$ for the function $f$ is given by the following Lemma.

\begin{lemma}
If $\tilde{f}: P \to \R$ is a $(k, k', \Delta)$-functional-sample of $f: \manifold \to \R$ with $k'\geq\frac{k}{2}$, then we have 
$| \newfunc{f}(p) - f(p) | \le \left(1+ 2 \sqrt{\frac{k - k'}{2k'-k}}\right) \Delta$ 
for every $p \in P$, where $\newfunc{f}$ is the disparity-based denoising of $\tilde{f}$. 
In particular, if $k' \ge \frac{2}{3}k$, then $| \newfunc{f}(p) - f(p) | \le 3 \Delta$ for every $p \in P$. 
\label{lem:discrep-bound}
\end{lemma}

\begin{proof}
Let $Y_\Delta = \{x\in \extNN_P^k(p) : |\tilde{f}(x)-f(p)| \leq \Delta \}$ be the set of points in $\extNN_P^k(p)$ whose 
observed function values are within distance $\Delta$ from $f(p)$. 
Since $\tilde{f}$ is a $(k, k', \Delta)$-functional-sample of $f$, it is clear that $|Y_\Delta| \geq k'$.
Let $Y'_{\Delta} \subset Y_\Delta$ be a subset with $k'$ elements, $Y'_{\Delta} = \{x'_i\}_{i=1}^{k'}$.
By the definitions of $Y_\Delta$ and $Y'_{\Delta}$, one can immediately check that 
$|\tilde{f}(x'_i)-\mu(Y'_{\Delta})| \leq 2\Delta$ where $\mu(Y'_{\Delta}) = \frac{1}{k'}\sum_{i=1}^{k'} \tilde{f}(x'_i)$.
This inequality then gives an upper bound of the disparity $\phi(Y'_{\Delta})$,
\[
\begin{array}{ccl}
\phi(Y'_{\Delta}) &= & \frac{1}{k'} \sum_{i=1}^{k'}(\tilde{f}(x'_i)-\mu(Y'_{\Delta}))^{2} \\
&\leq & \frac{1}{k'}\sum_{i=1}^{k'} (2\Delta)^{2} \\
&= &4 \Delta^{2}
\end{array}.
\]

Recall from the sliding window procedure that $\newfunc{Y}_p = \argmin_{Y_i, i\in [1, k-k']} \discrep(Y_i)$ and $\newfunc{z}_p = \mu(\newfunc{Y}_p)$.
Denote $A_1 = \newfunc{Y}_p \cap Y_\Delta$ and $A_2 = \newfunc{Y}_p \setminus A_1$.
Since $\tilde{f}$ is a $(k, k', \Delta)$-functional-sample of $f$, the size  of $A_2$ is at most $k-k'$ and $|A_1|\geq 2k'-k$.
If $|\newfunc{z}_p - f(p)| \leq \Delta$, 
nothing needs to be proved.
Without loss of generality, one can assume that $f(p) + \Delta \leq \newfunc{z}_p$.
Denote $\delta=\newfunc{z}_p - (f(p) + \Delta)$.
The disparity of $\phi(\newfunc{Y}_p)$ can then be estimated.
\[
\begin{array}{ccl}
\phi(\newfunc{Y}_p) & = & \frac{1}{k'}\left(\sum_{x\in A_1} (\tilde{f}(x)-\newfunc{z}_p)^{2} + \sum_{x\in A_2} (\tilde{f}(x)-\newfunc{z}_p)^{2}\right) \\
&\geq & \frac{1}{k'}\left(|A_1| \delta^{2} + \sum_{x\in A_2} (\tilde{f}(x)-\newfunc{z}_p)^{2}\right) \\
&\geq & \frac{1}{k'}\left(|A_1| \delta^{2} + \frac{1}{|A_2|}(\sum_{x\in A_2}\tilde{f}(x)-|A_2|\newfunc{z}_p)^{2}\right)\\
&=& \frac{1}{k'}\left(|A_1| \delta^{2} + \frac{1}{|A_2|}(\sum_{x\in A_1}\tilde{f}(x)-|A_1|\newfunc{z}_p)^{2}\right) \\
&\geq& \frac{1}{k'}\left(|A_1| \delta^{2} + \frac{1}{|A_2|}(|A_1|\delta)^{2}\right) \\
&=& \frac{1}{k'}\delta^{2}\left(\frac{|A_1|}{|A_2|}(|A_1|+|A_2|)\right) \\
&\geq& \frac{1}{k'}\delta^{2}\left(\frac{k'|A_1|}{|A_2|}\right) \\
&\geq& \frac{2k'-k}{k-k'}\delta^{2} \\
\end{array}
\]
where the third line uses the inequality $\sum_{i=1}^{n}a_i^{2} \geq \frac{1}{n}(\sum_{i=1}^{n}a_i)^{2}$, and the fourth line uses the fact that $(|A_1| + |A_2|)\newfunc{z}_p = \sum_{x\in \newfunc{Y}_p} \tilde{f}(x)$. 
Since $\newfunc{Y}_p = \argmin_{Y_i, i\in [1, k-k']} \discrep(Y_i)$, it holds that $\phi(\newfunc{Y}_p)  \leq \phi(Y'_{\Delta})$.
Therefore, 
\[
\frac{2k'-k}{k-k'}\delta^{2} \leq 4 \Delta^{2}.
\]
It then follows that $\delta \leq 2 \sqrt{\frac{k - k'}{2k'-k}} \Delta$
and $| \newfunc{f}(p) - f(p) | \le \left(1+ 2 \sqrt{\frac{k - k'}{2k'-k}}\right)\Delta $ since $\newfunc{z}_p = \newfunc{f}(p)$. 
If $k' \ge \frac{2}{3}k$, then $1+ 2 \sqrt{\frac{k - k'}{2k'-k}} \le 1 + 2 = 3$, meaning that $| \newfunc{f}(p) - f(p) | \le 3\Delta$ in this case. 

\end{proof}

\begin{corollary}
Given a $(k, k', \Delta)$-\funcsample{} of $f: \manifold \to \R$ with $k'\ge k/2$, we can compute a new function $\newfunc{f}: P \rightarrow \R$ such that $| \newfunc{f}(p) - f(p) | \le \fprecision\Delta$ for any $p\in P$, where $\fprecision = 1$ under $k$-median denoising, and $\fprecision = \left(1+ 2 \sqrt{\frac{k - k'}{2k'-k}}\right)$ under the disparity-based denoising. 
 
\label{cor:denoisingbound}
\end{corollary}

Hence after the $k$-median denoising or the disparity-based denoising, we obtain a new function $\newfunc{f}$ whose value at each sample point is within $\fprecision\Delta$ precision to the true function value. We can now apply the scalar field topology inference framework from \cite{sfapcdCGOS} (as introduced in Section \ref{sec:scalarpreliminary}) using $\hat f$ as input. 
In particular, set $L_\alpha = \{ p \in P \mid \newfunc{f}(p) \le \alpha \}$, and let $R_\delta( X )$ denote the Rips complex over points in $X$ with parameter $\delta$. We approximate the persistence diagram induced by the sub-level sets filtration of $f: \manifold \to \R$ from the filtrations of nested pairs $\{R_\delta(L_\alpha)\hookrightarrow R_{2\delta}(L_\alpha)\}_\alpha$. 
It follows from Theorem \ref{thm:funcnoiseonlyfromCGOS} that:
\begin{theorem}\label{tFuncNoise}
Let $\manifold$ be a compact Riemannian manifold and let $f:\manifold \rightarrow\R$ be a $c$-Lipschitz function.
Let $P$ be a geodesic $\eps$-sample of $\manifold$, and $\tilde{f}: P \rightarrow \R$ a $(k, k', \Delta)$-\funcsample{} of $f$. 
Set $\fprecision = 1$ if $P_\alpha$ is obtained via $k$-median denoising, and $\fprecision = \left(1+ 2 \sqrt{\frac{k - k'}{2k'-k}}\right)$ if $P_\alpha$ is obtained via disparity-based denoising. 
If $\eps<\frac{1}{4}\varrho(\manifold)$, then for any $\delta\in\left[2\eps,\frac{1}{2}\varrho(\manifold)\right)$, the persistent homology modules of $f$ and the filtration of nested pairs $\{R_\delta(P_\alpha)\hookrightarrow R_{2\delta}(P_\alpha)\}$ are $(2c\delta+\fprecision\Delta)$-interleaved.
Therefore, the bottleneck distance between their persistence diagrams is at most $2c\delta+\fprecision\Delta$. 
\label{thm:funcnoiseonly}
\end{theorem}


The above theoretical results are similar for $k$-median and disparity-based methods with a slight advantage for the $k$-median.
However, interesting experimental results can be obtained when the Lipschitz condition on the function is removed, for example with images, where the disparity based method appears to be more resilient to large amounts of noise than the $k$-median denoising method. 
Illustrating examples can be found in Appendix~\ref{sec:illustration}.

\section{Geometric noise}
\label{sec:geomonly}

In the previous section, we assumed that we have no geometric noise in the input.
In this section, we deal with the case where there is only 
geometric noise in the input, but no functional noise of any kind. 
Specifically, for any point $p\in P$, we assume that the observed value $\tilde f(p)$ is equal to the true function value $f(\pi(p))$ where $\pi(p)$ is the nearest point projection of $p$ to the manifold.
If $p$ is on the medial axis of $\manifold$, the projection $\pi$ is arbitrary to one of the nearest points.
As we have alluded before, general geometric noise implicitly introduces functional noise because the point $p$ may have become a functional  aberration of its orthogonal projection $\pi(p)\in \manifold$.
This error will be ultimately dealt with in Section \ref{sec:walkthrough} when we combine the results on purely functional noise 
from the previous section with the results 
on purely geometric noise in this section.

\subsection{Sampling condition} 
\label{sec:geomnoisemodel}
\paragraph*{Distance to a measure.}
The distance to a measure is a tool introduced to deal with geometrically noisy datasets, which are modelled as probability measures~\cite{gipmCCM}.
Given a probability measure $\mu$ on a metric space $\X$, we define the \emph{pseudo-distance} $\delta_m(x)$ for any point $x\in\mathbb{R}^d$ and a mass parameter $m\in(0,1]$ as 
$\delta_m(x)=\inf\{r\in\mathbb{R}|\mu(B(x,r))\geq m\}.$ 
The distance to a measure is then defined by averaging this quantity:
$$\dm(x)=\sqrt{\frac{1}{m}\int_0^{m}\delta_l(x)^2\ dl}.$$

The \emph{Wasserstein distance} is a standard tool to compare two measures.
Given two probability measures $\mu$ and $\nu$ on a metric space $\X$, a \emph{transport plan} $\pi$ is a probability measure over $\X\times\X$ such that for any $A\times B\subset \X\times \X$, $\pi(A\times \X)=\mu(A)$ and $\pi(\X\times B)=\nu(B)$.
Let $\Gamma(\mu,\nu)$ be the set of all transport plans between between measures $\mu$ and $\nu$.
The Wassserstein distance is then defined as the minimum transport cost over $\Gamma(\mu,\nu)$:
$$W_2(\mu,\nu)=\sqrt{\min_{\pi\in\Gamma(\mu,\nu)}\int_{\X\times\X}d_\X(x,y)^2 ~d\pi(x,y)},$$
where $d_\X(x,y)$ is the distance between $x$ and $y$ in the metric space $\X$. 
The distance to a measure is stable with respect to the Wasserstein distance as shown in \cite{gipmCCM}: 
\begin{theorem}[Theorem 3.5 of~\cite{gipmCCM}, Theorem 3.2 of~\cite{MEsparse}]\label{tStability}
Let $\mu$ and $\nu$ be two probability measures on $\X$ and $m\in(0,1]$. 
Then, $\norm{\dm-\dn}_\infty\leq\frac{1}{\sqrt{m}}W_2(\mu,\nu)$.
\end{theorem}

We will mainly use the distance to empirical measures in this paper. (See \cite{MEsparse,gipmCCM,wkdGMM} for more details on distance to a measure and its approximation.)
Given a finite point set $P$, its associated \emph{empirical measure $\mu_P$} is defined as the sum of Dirac masses:
$\mu_P=\frac{1}{|P|}\sum_{p\in P}\delta_p. $
The distance to this empirical measure for a point $x$ can then be expressed as an average of its distances to the $k = m|P|$ nearest neighbors where $m$ is the mass parameter. 
For the sake of simplicity, $k$ will be assumed to be an integer.
The results also hold for other values of $k$.
However, a non integer $k$ introduces unnecessary technical difficulties.
Denoting by $p_i(x)$ the $i$-th nearest neighbors of $x$ in $P$, one can write:
$$\dmP(x)=\sqrt{\frac{1}{k}\sum_{i=1}^k d(p_i(x),x)^2}.$$

\paragraph*{Geometric sampling condition.}
Our sampling condition treats the input
point data as a measure and relates it to the manifold (where input points are sampled from) via distance-to-measures with the help of two parameters. 

\begin{definition}\label{def:geomnoise}
Let $P\subset\mathbb{R}^n$ be a discrete sample and $\manifold\subset\R^n$ a smooth manifold.
Let $\mu_P$ denote the empirical measure
of $P$. For a fixed mass parameter $m>0$, we say that $P$ is an $(\eps,r)$-sample of 
$\manifold$ if the following holds:
\begin{eqnarray}\label{dense}
\forall x \in \manifold, \dmP(x) \le \eps; ~~\text{and} 
\end{eqnarray}
\begin{eqnarray}\label{sparse}
\forall x \in \R^n,\ \dmP(x)\leq r \implies d(x,\manifold)\leq \dmP(x)+\eps. 
\end{eqnarray}
\end{definition}
The parameter $\eps$ captures 
the distance to the empirical measure for points in $\manifold$ and intuitively 
tells us how dense $P$ is in relation  
to the manifold $\manifold$. 
The parameter $r$ intuitively indicates how far away we can deviate
from the manifold, while keeping the noise
sparse enough so as not to be mistaken for signal.
We remark that if a point set is an $(\eps,r)$-sample of $\manifold$ then it is an $(\eps',r')$-sample of $\manifold$ for any $\eps'\geq\eps$ and $r'\leq r$. In general, the smaller $\eps$ is and the bigger $r$ is, the better an $(\eps,r)$-sample is. 

For convenience, denote the distance
function to the manifold $\manifold$ by  
$d_\pi: \mathbb{R}^n\rightarrow \mathbb{R}$, $x\mapsto d(x,\manifold)$.
We have the following interleaving relation: 
\begin{equation}\label{eqn:interleave}
\forall \alpha<r-\eps,\ d_\pi^{-1}(]-\infty,\alpha])\subset\dmP^{-1}(]-\infty,\alpha+\eps])\subset d_\pi^{-1}(]-\infty,\alpha+2\eps])
\end{equation}

To see why this interleaving relation holds, let $x$ be a point such that $d(x,\manifold)\leq\alpha$. 
Thus $d(\pi(x),x)\leq\alpha$. 
Using the hypothesis (\ref{dense}), we get that $\dmP(\pi(x))\leq\eps$.
Given that the distance to a measure is a 1-Lipschitz function we then obtain that $\dmP(x)\leq\eps+\alpha$.

Now let $x$ be a point such that $\dmP(x)\leq\alpha+\eps\leq r$.
Using the condition
on $r$ in (\ref{sparse}) we get that $d(x,\manifold)\leq \dmP(x)+\eps\leq \alpha+2\eps$ which concludes the proof of  Eqn~(\ref{eqn:interleave}).

Eqn (\ref{eqn:interleave}) gives an interleaving between the sub-level 
sets of the distance to the measure $\mu$ and the offsets of the manifold $\manifold$.
By Theorem~\ref{tstability}, this implies the proximity between the persistence modules of their respective  
sub-level sets filtrations .
Observe that this relation is in some sense analogous to the one obtained 
when two compact sets $A$ and $B$ have Hasudorff distance of at most 
$\eps$:
\begin{equation}\label{eqn:classicalHaus}
\forall \alpha,\ d_A^{-1}(]-\infty,\alpha])\subset d_B^{-1}(]-\infty,\alpha+\eps])\subset d_A^{-1}(]-\infty,\alpha+2\eps]) . 
\end{equation}

\paragraph*{Relation to other sampling conditions.}
Our sampling condition encompasses several other existing sampling conditions.
While the parameter $\eps$ is natural, the parameter $r$ may appear to be artificial.
It bounds the distances at which we can observe the manifold 
through the scope of the distance to a measure.
In most classical sampling conditions, $r$ is equal to $\infty$ and thus 
we obtain a similar relation as for the classical Hausdorff sampling condition in Eqn (\ref{eqn:classicalHaus}). 

One notable noise model where $r\neq\infty$ is when there is an uniform background noise in the ambient space $\R^d$, sometimes called \emph{clutter noise}.
In this case, $r$ depends on the difference between the density of the relevant data and the density of the noise.
For other sampling conditions like Wassertein, Gaussian, Hausdorff sampling conditions, $r=\infty$.
Detailed relations and proofs for the Wasserstein and Gaussian sampling conditions can be found in Appendix~\ref{sec:relNoiseGeom}.

\subsection{Scalar field analysis under geometric noise}
In the rest of the paper,  we assume that $\manifold$ is a manifold with 
positive reach $\rho_\manifold$ (minimum distance between
$\manifold$ and its medial axis) and 
whose curvature is bounded by $c_\manifold$.
Assume that the input $P$  is an $(\eps,r)$-sample of 
$\manifold$ for a given $m\in(0,1]$, where 
\begin{equation}\label{dense2}
\eps\leq \frac{\rho_\manifold}{6}\mbox{ , } \mbox{ and } r>2\eps. 
\end{equation}
As discussed at the beginning of this section, we assume that there is no intrinsic functional noise, that is, for every $p \in P$, the observed function value $\tilde f(p) = f(\pi(p))$ is the same as the true value for the projection $\pi(p) \in \manifold$ of this point. 
Our goal now is to show how to recover the persistence diagram induced by $f: \manifold \to\R$ from its observations $\tilde f: P \to \R$ on $P$.

Taking advantage of the interleaving~(\ref{eqn:interleave}), we can use the distance to the empirical measure to filter the points of $P$ to remove geometric noise. 
In particular, we consider the set 
\begin{equation}\label{eqn:L}
L=P\cap \dmP^{-1}(]-\infty,\eta]) \mbox{ where } \eta\geq2\eps.
\end{equation}
We will then use a similar approach as the one from~\cite{sfapcdCGOS} for this set $L$. 
The optimal choice for the parameter $\eta$ is $2\eps$. 
However, any value with $\eta\leq r$ and $\eta+\eps< \rho_\manifold$ works as long as there exist $\delta$ and $\delta'$ satisfying the conditions stated in Theorem~\ref{tGeomNoise}. 

Let $\bar{L}=\{\pi(x)|x\in L\}$ denote the orthogonal projection of $L$ onto $\manifold$. 
To simulate sub-level sets $f^{-1}(]-\infty,\alpha]$ of $f: \manifold \to \R$, consider the restricted sets  
$L_\alpha := L\cap(f\circ\pi)^{-1}(]-\infty,\alpha])$ and 
let $\bar{L}_\alpha=\pi(L_\alpha)$.
By our assumption on the observed function $\tilde f: P \to \R$, we have: 
$L_\alpha=\{x\in L|\tilde f(x)\leq\alpha\}$.

Let us first recall a result about the relation between Riemannian and 
Euclidian metrics (e.g.~\cite{DSW11}).
For any two points
$x, y\in \manifold$ with $d(x,y)\leq\frac{\rho_\manifold}{2}$
one has:
\begin{equation}\label{eucriem}
d(x,y)\leq d_\manifold(x,y)\leq\left(1+\frac{4d(x,y)^2}{3\rho_\manifold^2}\right)d(x,y)\leq \frac{4}{3}d(x,y).
\end{equation}


As a direct consequence of our sampling condition, for each point $x\in\manifold$, there exists a point $p\in L$ at distance less than $2\eps$: Indeed, for each $x\in \manifold$, since $\dmP(x) \le \eps$, there must exist a point $p\in P$ such that $d(x,p) \le \eps$. On the other hand, since the distance to measure is $1$-Lipschitz, we have $\dmP(p) \le \dmP(x) + d(x,p) \le 2\eps$. Hence $p \in L$ as long as $\eta\geq2\eps$.
We will use the \emph{extrinsic} Vietoris-Rips complex built on top of points from $L$ 
to infer the scalar field topology. 
Using the previous relation Eqn (\ref{eucriem}), we obtain the following result which states that the Euclidean distance for nearby points in $L$ approximates the geodesic distance on $\manifold$. 
\begin{proposition}
Let $\lambda=\frac{4}{3}\frac{\rho_\manifold}{\rho_\manifold-(\eta+\eps)}$, and assume that $2\eps \leq \eta \leq r$ and $\eps+\eta<\rho_\manifold$. 
Let $x,y \in L$ be two points from $L$ such that $d(x,y)\leq \frac{\rho_\manifold}{2}-\frac{\eta+\eps}{2}$. 
Then, 
$$\frac{d_\manifold(\pi(y),\pi(x))}{\lambda}\leq d(x,y)\leq 2(\eta+\eps) + d_\manifold(\pi(x),\pi(y)). $$
\label{prop:distance}
\end{proposition}

\begin{proof}
Let $x$ and $y$ be two points of $L$ such that $d(x,y)\leq \frac{\rho_\manifold}{2}-\frac{\eta+\eps}{2}$. 
As $\dmP(x)\leq\eta\leq r$, Eqn (\ref{sparse}) implies $d(\pi(x),x)\leq \eta+\eps$.
Therefore, $d(\pi(x),\pi(y))\leq\frac{\rho_\manifold}{\rho_\manifold-(\eta+\eps)}d(x,y)$~\cite[Theorem 4.8,(8)]{cmF}.
This implies $d(\pi(x),\pi(y))\leq\frac{\rho_\manifold}{2}$ and following~(\ref{eucriem}), $d_\manifold(\pi(x),\pi(y))\leq\frac{4}{3}d(\pi(x),\pi(y))$.

This proves the left inequality in the Proposition. The right inequality follows from $$d(x,y)\leq d(\pi(x),x)+d(\pi(y),y)+d_\manifold(\pi(x),\pi(y))
\leq 2(\eta+\eps)+d_\manifold(\pi(x),\pi(y)).$$
\end{proof}

\begin{theorem}\label{tGeomNoiseNew}
Let $\manifold$ be a compact Riemannian manifold and let $f:\manifold \rightarrow\R$ be a $c$-Lipschitz function. Let $P$ be an $(\eps,r)$-sample of $M$, and $L$ be as introduced in Eqn (\ref{eqn:L}). Assume $\eps \le \frac{\rho_\manifold}{6}, r > 2\eps$, and $2\eps \le \eta \leq r$. 
Then, for any $\delta\geq2\eta+6\eps$ and any $\delta'\in\left[2\eta+2\eps+\frac{8}{3}\frac{\rho_\manifold}{\rho_\manifold-(\eta+\eps)}\delta,\ \frac{3}{4}\frac{\rho_\manifold-(\eta+\eps)}{\rho_\manifold}\varrho(\manifold)\right]$, $H_*(f)$ and 
$H_*(R_\delta({L}_\alpha)\hookrightarrow R_{\delta'}({L}_\alpha))$ are $\frac{4}{3}\frac{c \rho_\manifold \delta'}{\rho_\manifold-(\eta+\eps)}$-interleaved.
\end{theorem}

\begin{proof}
First, note that $\bar{L}$ is a $2\eps$-sample of $\manifold$ in its geodesic metric. It follows from the definition of $\dmP$ that,
for any point $x\in \manifold$, the nearest point $p\in L$ to $x$
satisfies $d(x,p) \le\dmP(x)\le \eps$. Hence $d(x, \pi(p)) \le d(x, p) + d(p, \pi(p)) \le 2d(x,p) \le 2\eps$. 
Now we apply Theorem \ref{tGeomNoise} to $\bar{L}$ by using $\tilde d(\pi(x), \pi(y)) := d(x,y)$; and setting $\lambda=\mu=\frac{4}{3}\frac{\rho_\manifold}{\rho_\manifold-(\eta+\eps)}$, $\nu=2(\eta+\eps)$: the requirement on the distance function $\tilde d$ in Theorem \ref{tGeomNoise} is satisfied due to Proposition \ref{prop:distance}. The claim then follows. 
\end{proof}

%

Since $\manifold$ is compact, $f$ is bounded due to the Lipschitz condition.
We can look at the limit when $\alpha\rightarrow\infty$.
There exists a value $T$ such that for any $\alpha\geq T$, $L_\alpha=L$ and $f^{-1}((-\infty,\alpha])=\manifold$.
The above interleaving means that $H_*(\manifold)$ and $ H_*(R_\delta(L))\hookrightarrow R_{\delta'}(L))$ are interleaved.
However, both objects do not depend on $\alpha$ and this gives the following inference result:
\begin{corollary}
$H_*(\manifold)$ and $H_*(R_\delta(L))\hookrightarrow R_{\delta'}(L))$ are isomorphic under conditions specified in Theorem \ref{tGeomNoiseNew}.
\end{corollary}

\section{Scalar Field Topology Inference under Geometric and Functional Noise}
\label{sec:walkthrough}


Our constructions can be combined to analyze scalar fields in a 
more realistic setting.
Our \emph{combined sampling condition} follows conditions~(\ref{dense}) and~(\ref{sparse}) for the geometry.
We adapt condition~(\ref{eqn:funcnoisemodel}) 
to take into account the geometry and introduce the following conditions: we assume that there 
exist $\eta \ge 2\eps$ and $s$ such that:
\begin{equation}\label{hFinal3}
\forall p\in \dm^{-1}((-\infty,\eta,]),\ |\{q\in NN_k(p)|\ |\tilde f(q)-f(\pi(p))|\leq s\}|\geq k'
\end{equation} 

Note that in~(\ref{hFinal3}), we are using $f(\pi(p))$ as the ``true" function value at a sample $p$ which may be off the manifold $\manifold$. The condition on the functional noise is only for points close to the manifold (under the distance to a measure). 
Combining the methods from the previous two sections,
we obtain the \emph{combined noise algorithm} where $\eta$ is a parameter greater than $2\eps$.

We propose the following 3-steps algortihm.
It starts by handling outliers in the geometry then it makes a regression on the function values to obtain a smoothed function $\hat{f}$ before running the existing algorithm for scalar field analysis~\cite{sfapcdCGOS} on the filtration $\hat{L}_\alpha=\{p\in L|\hat{f}(p)\leq\alpha\}$.
\begin{figure}[!h]
\begin{algorithm}{Combined noise algorithm}
\begin{enumerate}
\item Compute $L=P\cap\dm^{-1}((-\infty,\eta])$.
\item Replace functional values $\tilde f$ by $\hat f$ for points in $L$ using either k-median or disparity based method.
\item Run the scalar field analysis algorithm from~\cite{sfapcdCGOS} on $(L,\hat f)$.
\end{enumerate}
\end{algorithm}
\end{figure}

\begin{theorem}\label{tCombinedNoise}
Let $\manifold$ be a compact smooth manifold embedded in $\R^d$ and $f$ a $c$-Lipschitz function on $\manifold$.
Let $P \subset \R^d$ be a point set and $\tilde f: P \to \R$ be observed function values such that hypotheses~(\ref{dense}),~(\ref{sparse}),~(\ref{dense2}) and~(\ref{hFinal3}) are satisfied.
For $\eta\geq 2\eps$, the combined noise algorithm has the following guarantees: 

For any $\delta\in\left[2\eta+6\eps,\frac{\varrho(\manifold)}{2}\right]$ and any $\delta'\in\left[2\eta+2\eps+\frac{8}{3}\frac{\rho_\manifold}{\rho_\manifold-(\eta+\eps)}\delta,\frac{3}{4}\frac{\rho_\manifold-(\eta+\eps)}{\rho_\manifold}\varrho(\manifold)\right]$, $H_*(f)$ and $H_*({R}_\delta(\hat{L}_\alpha)\hookrightarrow{R}_{\delta'}(\hat{L}_\alpha))$ are $\left(\frac{4}{3}\frac{c\rho_\manifold\delta'}{\rho_\manifold-(\eta+\eps)}+\fprecision s\right)$-interleaved where $\fprecision=1$ if we use the $k$-median and $\fprecision=\left(1+ 2 \sqrt{\frac{k - k'}{2k'-k}}\right)$ if we use the disparity method for Step 2.
\end{theorem}


\begin{proof}
First, consider the filtration induced by $L_\alpha=\{x\in L|f(\pi(x))\leq\alpha\}$; that is, we first imagine that all points in $L$ have correct function values (equals to the true value of their projection on $\manifold$). By Theorem \ref{tGeomNoiseNew}, for  $$\delta\in\left[2\eta+6\eps,\frac{\varrho(\manifold)}{2}\right] \text{ and } \delta'\in\left[2\eta+2\eps+\frac{8}{3}\frac{\rho_\manifold}{\rho_\manifold-(\eta+\eps)}\delta,\frac{3}{4}\frac{\rho_\manifold-(\eta+\eps)}{\rho_\manifold}\varrho(\manifold)\right],$$ $H_*(f)$ and $H_*({R}_\delta(L_\alpha)\hookrightarrow{R}_{\delta'}(L_\alpha))$ are $\frac{4}{3}\frac{c \rho_\manifold\delta'}{\rho_\manifold-(\eta+\eps)}$-interleaved. 

Next, consider $\hat L_\alpha=\{p\in L|\hat f(p)\leq\alpha\}$, which leads to a filtration based on the smoothed function values $\hat{f}$ (not observed values). 
Recall that our algorithm returns $H_*({R}_\delta(\hat{L}_\alpha)\hookrightarrow{R}_{\delta'}(\hat{L}_\alpha))$. We aim to relate this persistence module with $H_*({R}_\delta(L_\alpha)\hookrightarrow{R}_{\delta'}(L_\alpha))$. 
Specifically, fix $\alpha$ and let $(x,y)$ be an an edge of $R_\delta(L_\alpha)$.
This means that $d(x,y)\leq 2\delta$, $f(\pi(x))\leq\alpha$, $f(\pi(y))\leq\alpha$.
Corollary~\ref{cor:denoisingbound} can be applied to the function $f\circ\pi$ due to hypothesis~(\ref{hFinal3}).
Hence $|\hat f(x)-f(\pi(x))|\leq\fprecision s$ and $|\hat f(y)-f(\pi(y))|\leq\fprecision s$.
Thus $(x,y)\in R_\delta(\hat L_{\alpha+\fprecision s})$.
One can reverse the role of $\hat f$ and $f$ and get an $\fprecision s$-interleaving of $\{R_\delta(L_\alpha)\}$ and $\{R_\delta(\hat L_\alpha)\}$. 
This gives rise to the following commutative diagram since all arrows are induced by inclusions. 
\begin{center}
\begin{tikzpicture}[scale=.7]
\draw (0,0) node {$H_*(R_{\delta}(L_\alpha))$};
\draw (5,0) node {$H_*(R_{\delta}(L_{\alpha+2\fprecision s}))$};
\draw (10,0) node{$H_*(R_{\delta}(L_{\alpha+4\fprecision s}))$};

\draw (2.5,2) node {$H_*(R_{\delta}(\hat L_{\alpha+\fprecision s}))$};
\draw (7.5,2) node {$H_*(R_{\delta}(\hat L_{\alpha+3\fprecision s}))$};
\draw (12.5,2) node {$H_*(R_{\delta}(\hat L_{\alpha+5\fprecision s}))$};

\draw[->] (2,0) -- (3,0);
\draw[->] (7,0) -- (8,0);

\draw[->] (4.5,2) -- (5.5,2);
\draw[->] (9.5,2) -- (10.5,2);

\draw[->] (.3,.7) -- (2.2,1.3);
\draw[->] (2.8,1.3) -- (4.7,.7);
\draw[->] (5.3,.7) -- (7.2,1.3);
\draw[->] (7.8,1.3) -- (9.7,.7);
\draw[->] (10.3,.7) -- (12.2,1.3);

\draw (0,6) node {$H_*(R_{\delta'}(L_\alpha))$};
\draw (5,6) node {$H_*(R_{\delta'}(L_{\alpha+2\fprecision s}))$};
\draw (10,6) node{$H_*(R_{\delta'}(L_{\alpha+4\fprecision s}))$};

\draw (2.5,8) node {$H_*(R_{\delta'}(\hat L_{\alpha+\fprecision s}))$};
\draw (7.5,8) node {$H_*(R_{\delta'}(\hat L_{\alpha+3\fprecision s}))$};
\draw (12.5,8) node {$H_*(R_{\delta'}(\hat L_{\alpha+5\fprecision s}))$};

\draw[->] (2,6) -- (3,6);
\draw[->] (7,6) -- (8,6);

\draw[->] (4.5,8) -- (5.5,8);
\draw[->] (9.5,8) -- (10.5,8);

\draw[->] (.3,6.7) -- (2.2,7.3);
\draw[->] (2.8,7.3) -- (4.7,6.7);
\draw[->] (5.3,6.7) -- (7.2,7.3);
\draw[->] (7.8,7.3) -- (9.7,6.7);
\draw[->] (10.3,6.7) -- (12.2,7.3);

\draw[->] (0,.7) -- (0,5.3);
\draw (5,.7) -- (5,1.9);
\draw[->] (5,2.1) -- (5,5.3);
\draw (10,.7) -- (10,1.9);
\draw[->] (10,2.1) -- (10,5.3);
\draw (2.5,2.7) -- (2.5,5.9);
\draw[->] (2.5,6.1) -- (2.5,7.3);
\draw (7.5,2.7) -- (7.5,5.9);
\draw[->] (7.5,6.1) -- (7.5,7.3);
\draw[->] (12.5,2.7) -- (12.5,7.3);
\end{tikzpicture}
\end{center}
Thus the two persistence modules induced by filtrations of nested pairs $\{{R}_\delta(L_\alpha)\hookrightarrow{R}_{\delta'}(L_\alpha)\}$ and $\{{R}_\delta(\hat L_\alpha)\hookrightarrow{R}_{\delta'}(\hat L_\alpha)\}$ are $\fprecision s$-interleaved.  
Combining this with the interleaving between $H_*({R}_\delta(L_\alpha)\hookrightarrow{R}_{\delta'}(L_\alpha))$ and $H_*(f)$, we obtain the stated results. 
\end{proof}

We note that, while this theorem assumes a setting where we can ensure theoretical guarantees, the algorithm can be applied in a more general setting still producing good results.

\paragraph*{Acknowledgments}

This work was supported by the ANR project TopData 13-BS01-008, the ERC project Gudhi 339025 and the NSF grants CCF-1064416, CCF-1116258, CCF-1319406 and CCF-1318595.

\newpage
\bibliography{arxiv}

\newpage
\appendix
\section{Relations between our functional sampling condition and classical noise models}\label{sec:relNoiseFunc}
\paragraph*{Bounded noise model.}
The standard ``bounded noise'' model assumes that all observed function values are within some $\delta$ distance away from the true function values: that is, $|\tilde f(p) - f(p) | \le \delta$ for all $p \in P$. 
Hence this bounded noise model simply corresponds to a $(1, 1, \delta)$-\funcsample{}.

\paragraph*{Gaussian noise model.}
Under the popular Gaussian noise model, for any $x\in \manifold$, its observed function value $\tilde f(x)$ is drawn from a normal distribution $\mathcal{N}(f(x), \sigma)$, that is a probability measure with density $g(y)=\frac{1}{\sigma \sqrt{\pi}} e^{-\frac{(y - f(x))^2}{\sigma^2}}$. 
We say that a point $q \in P$ is $a$-accurate if $| \tilde{f}(q) - f(q)| \le a$. 
For the Gaussian noise model, we will first bound the quantity $\mu(k, k')$ defined as the smallest value such that at least $k'$ out of the $k$ nearest neighbors of $p$ in $\extNN_P^k(p)$ are $\mu(k, k')$-accurate. 
We claim the following statement. 
\begin{claim}
With probability at least $1 - e^{-\frac{k-k'}{6}}$, $\mu(k, k') \le \sigma \sqrt{\ln \frac{2k}{k-k'}}$. 
\label{claim:gaussianaccuracy}
\end{claim}
\begin{proof}
First note that for $\frac{b}{\sigma} \geq 1$, we have that: 
$$\int_b^{+\infty}e^{-\frac{t^2}{\sigma^2}} dt \le \int_b^{+\infty} \frac{t}{\sigma} e^{-\frac{t^2}{\sigma^2}} dt = \frac{1}{\sigma} \int_b^{+\infty} t e^{-\frac{t^2}{\sigma^2}} dt = - \frac{\sigma}{2} e^{-\frac{t^2}{\sigma^2}} \big|_{b}^{\infty} = \frac{\sigma}{2} e^{-\frac{b^2}{\sigma^2}}. 
$$
Now we introduce $I(a)= \frac{1}{\sigma \sqrt{\pi}} \int_{-a}^a e^{-\frac{x^2}{\sigma^2}} dx$. 
Since $\frac{1}{\sigma\sqrt{\pi}} \int_{-\infty}^{\infty} e^{-\frac{x^2}{\sigma^2}} dx= 1$, we thus obtain that for $a\geq\sigma$:
\begin{equation}
1 - \frac{1}{\sqrt{\pi}}e^{-(\frac{a}{\sigma})^2} < 1 - e^{-(\frac{a}{\sigma})^2} \le I(a) ~(= 1 - \frac{2}{\sigma \sqrt{\pi}} \int_a^{+\infty}e^{-\frac{x^2}{\sigma^2}} dx).  
\label{eqn:boundGaussian}
\end{equation}

Now set $\delta = \frac{k-k'}{k}\leq\frac{1}{2}$ and $s = \sigma \sqrt{\ln \frac{2k}{k-k'}}\geq\sigma$. 
Let $p_1, \ldots, p_k$ denote the $k$ nearest neighbors of some point, say $p_1$. 
For each $p_i$, let $Z_i = 1$ if $p_i$ is {\bf not} $s$-\mygood{}, and $Z_i = 0$ otherwise. 
Hence $Z = \sum_{i=1}^k Z_i$ denotes the total number of points from these $k$ nearest neighbors that are not $s$-\mygood{}. 
By Equation (\ref{eqn:boundGaussian}), we know that 
$${\mathrm{Prob}} [ Z_i = 1 ] = 1 - I(s) \le e^{-(\frac{s}{\sigma})^2}. $$
It then follows that the expected value of $Z$ satisfies:
$$E(Z) \le ke^{-(\frac{s}{\sigma})^2} = \frac{\delta k}{2}. $$
Now set $\rho = \frac{\delta k}{2E(Z)}$. Since $E(Z) \le \frac{\delta k}{2}$, it follows that $(1+\rho)E(Z) \le \delta k$. 
Using Chernoff's bound~\cite{AV79}, we obtain
\begin{align*}
\mathrm{Prob}~[Z \ge k - k'] &= \mathrm{Prob}~[Z \ge \delta k] \le \mathrm{Prob}~[Z \ge (1+\rho) E(Z)] \\
& \le e^{-\frac{\rho^2 E(Z)}{2+\rho}} = e^{-\frac{\delta^2k^2}{4 E(Z)}\cdot \frac{1}{2+\frac{\delta k}{2E(Z)}}} \le e^{-\frac{\delta^2 k^2}{6\delta k}}= e^{-\frac{k-k'}{6}}. 
\end{align*}
The claim then follows, that is, with probability at least $1- e^{-\frac{k-k'}{6}}$, at least $k'$ number of points out of any $k$ points are $s = \sigma \sqrt{\ln \frac{2k}{k-k'}}\geq\sigma$-accurate.
\end{proof}

Next, we convert the value $\mu(k, k')$ to the value $\Delta$ as in Equation (\ref{eqn:funcnoisemodel}). 
In particular, being a $(k,  k', \Delta)$-\funcsample{} means that for any $p \in P$, there are at least $k'$ samples $q$ from $\extNN_P^k(p)$ such that $|\tilde{f}(q) - f(p) | \le \Delta$. 
Now assume that the furthest geodesic distance from any point in $\extNN_P^k(p)$ to $p$ is $\lambda$. 
Then since $f$ is a $c$-Lipschitz function, we have $\max_{q\in \extNN_P^k(p)} |f(q) - f(p)| \le c \lambda$.

We note that Claim \ref{claim:gaussianaccuracy} is valid for any point $p$ of $P$.
Using the union bound, the relation holds for all points in $P$ with probability at least $1 - ne^{-\frac{k-k'}{6}}$. 
Note that if $k - k' \ge 12 \ln n$, then this probability is at least $1 - \frac{1}{n}$, that is, the relation holds with high probability.  
Thus, with probability at least $1 - ne^{-\frac{k-k'}{6}}$, the input function $\tilde{f} :P \rightarrow \R$ under Gaussian noise model is a $(k, k', \Delta)$-\funcsample{} with $\Delta = \sigma \sqrt{\ln \frac{2k}{k-k'}} + c \lambda$. 



\section{Relations between our geometric sampling condition and the Wasserstein sampling condition}\label{sec:relNoiseGeom}

The Wasserstein sampling condition assumes that the empirical measure $\mu = \mu_P$ for $P$ is close to the uniform measure $\mu_\manifold$ on $\manifold$ under the Wasserstein distance. 
Let $\manifold$ be a $d'$-Riemannian manifold whose curvature is bounded from above by $c_\manifold$ and has a positive strong convexity radius $\varrho(\manifold)$.
Let $V_\manifold$ denote the volume of $\manifold$.  Writing, $\Gamma$ the Gamma function, let us set ${\cal C}_{d'}^{c_\manifold}$ to be the following constant: 
\begin{equation}\label{eqn:Cdm}
{\cal C}_{d'}^{c_\manifold}=\frac{4}{d'}\Gamma\left(\frac{1}{2}\right)^{d'}\Gamma\left(\frac{d'}{2}\right)^{-1}\left(\frac{\sqrt{c_\manifold}}{\pi}\right)^{d'-1}, 
\end{equation}
\begin{theorem}\label{tNoiseModelWas}
Let $P$ be a set of points whose empirical measure $\mu$ satisfies $W_2(\mu,\mu_\manifold)\leq\sigma$, where $\mu_\manifold$ is the uniform measure on $\manifold$. Then, for any $m\leq\frac{{\cal C}^{c_\manifold}_{d'}\left(\frac{\pi}{c_\manifold}\right)^{d'}}{V_\manifold}$, $P$ is an $(\eps, r)$-sample under our sampling condition for 
$$\eps\geq\frac{1}{\sqrt{1+\frac{2}{d'}}}\left(\frac{m V_\manifold}{{\cal C}_{d'}^{c_\manifold}}\right)^\frac{1}{d'}+\frac{\sigma}{\sqrt{m}},~~\mbox{ and }~~r=\infty.$$
\end{theorem}
\begin{proof}
Fixing a point $x\in\manifold$, we can lower bound the volume of the Riemannian ball of radius $a$, centered at $x$, using the Günther-Bishop Theorem:

\begin{theorem}[G\"unther-Bishop]\index{G\"unther-Bishop theorem}
Assuming that the sectional curvature of a manifold $\manifold$ is always less than $c_\manifold$ and $a$ is less than the strong convexity radius of $\manifold$, then for any point $x\in \manifold$, the volume $\cV(x,a)$ of the geodesic ball centred on $x$ and of radius $a$ is greater than $V_{d'}^{c_\manifold}(a)$ where $d'$ is the intrinsic dimension of $\manifold$ and $V_{d'}^{c_\manifold}(a)$ is the volume of the Riemannian ball of radius $a$ on a surface with constant curvature $c_\manifold$.
\end{theorem}

We explicitly bound the value of $\cV(x,a),$ with the following technical lemma:

\begin{lemma}\label{lem:volumeRiemBall}
Let $\manifold$ be a Riemannian manifold with curvature upper bounded by $c_\manifold$, then for any $x\in\manifold$ and $a\leq\min(\varrho(\manifold);\frac{\pi}{\sqrt{c_\manifold}})$, the volume ${\cal V}(x,a)$ of the geodesic ball centred at $x$ and of radius $a$ verifies:
$${\cal V}(x,a)\geq{\cal C}_{d'}^{c_\manifold}a^{d'}$$
where ${\cal C}_{d'}^{c_\manifold}$ is a constant independent of $x$ and $a$. 
\end{lemma}
\begin{proof}
Given $a\leq\min(\varrho(\manifold),\frac{\pi}{\sqrt{c_\manifold}})$, we want to bound the volume $V_{d'}^{c_\manifold}(a)$.
Consider the sphere of dimension $d'$ and curvature $c_\manifold$.
The surface $S_{c_\manifold}^{d'-1}$ of the border of a ball of radius $a\leq\frac{\pi}{\sqrt{c_\manifold}}$ on this sphere is given by~\cite{vsgbrmG}:
$$S^{d'-1}_{c_\manifold}(a)=2\Gamma\left(\frac{1}{2}\right)^{d'}\Gamma\left(\frac{d'}{2}\right)^{-1}c_\manifold^{-\frac{1}{2}(d'-1)}\sin^{d'-1}(c_\manifold a)$$
We can bound the value of $V_{d'}^{c_\manifold}(a)$ :
\begin{align*}V_{d'}^{c_\manifold}(a)&=\int_0^a S^{d'-1}(l)dl\\
&=\int_0^a 2\Gamma\left(\frac{1}{2}\right)^{d'}\Gamma\left(\frac{d'}{2}\right)^{-1}c_\manifold^{-\frac{1}{2}(d'-1)}\sin^{d'-1}(c_\manifold l)dl\\
&\geq 2\Gamma\left(\frac{1}{2}\right)^{d'}\Gamma\left(\frac{d'}{2}\right)^{-1}c_\manifold^{-\frac{1}{2}(d'-1)}2\int_0^{\frac{a}{2}} \left(\frac{2c_\manifold l}{\pi}\right)^{d'-1}dl\\
&=4\Gamma\left(\frac{1}{2}\right)^{d'}\Gamma\left(\frac{d'}{2}\right)^{-1}c_\manifold^{-\frac{1}{2}(d'-1)}\frac{\pi}{2c_\manifold}\int_0^{\frac{c_\manifold a}{\pi}}u^{d'-1} du
\end{align*}
Writing
$${\cal C}_{d'}^{c_\manifold}=\frac{4}{d'}\Gamma\left(\frac{1}{2}\right)^{d'}\Gamma\left(\frac{d'}{2}\right)^{-1}\left(\frac{\sqrt{c_\manifold}}{\pi}\right)^{d'-1},$$
and using the Günther-Bishop Theorem, we have for any $a\leq\min(\varrho(\manifold);\frac{\pi}{\sqrt{c_\manifold}})$ and any $x\in\manifold$,
$${\cal V}(x,a)\geq{\cal C}_{d'}^{c_\manifold}a^{d'}.$$
\end{proof}

We next prove that the empirical measure $\mu$ of $P$ satisfies the two conditions in Eqns (\ref{dense}) and (\ref{sparse}) for the value of $\eps$ and $r$ specified in Theorem \ref{tNoiseModelWas}.  
Specifically, recall that $\mu_\manifold$ be the uniform measure on $\manifold$ and $\mu$ is a measure such that $W_2(\mu,\mu_\manifold)\leq\sigma$.
Now  consider a point $x\in\manifold$ and the Euclidean ball $B(x,a)$ centred in $x$ and of radius $a$. 
By definition of $\mu_\manifold$, for any $a\leq\frac{\pi}{c_\manifold}$:
$$\mu_\manifold(B(x,a))=\frac{{\cal V}ol(x,a)}{V_\manifold}\geq\frac{{\cal C}_{d'}^{c_\manifold}a^{d'}}{V_\manifold}$$
By the definition of the pseudo-distance $\delta_m(x)$, we can then bound it, for any $m\leq\frac{{\cal C}_{d'}^{c_\manifold}\left(\frac{\pi}{c_\manifold}\right)^{d'}}{V_\manifold}$, as follows: 
$$\delta_m(x)\leq\left(\frac{m\ V_\manifold}{{\cal C}_{d'}^{c_\manifold}}\right)^{\frac{1}{d'}}.$$
This in turn produces an upper bound on the distance to the measure $\mu_\manifold$: 
\begin{align*}
\dM(x)&\leq\frac{1}{\sqrt{m}}\sqrt{\int_{0}^{m}\left(\frac{V_\manifold}{{\cal C}_{d'}^{c_\manifold}}l\right)^{\frac{2}{d'}} dl} 
~\leq\frac{1}{\sqrt{1+\frac{2}{d'}}}\left(\frac{V_\manifold m}{{\cal C}_{d'}^{c_\manifold}}\right)^\frac{1}{d'}
\end{align*}
By Theorem~\ref{tStability}, it then follows that for any $x \in \manifold$: 
$$\dm(x)\leq\frac{1}{\sqrt{1+\frac{2}{d'}}}\left(\frac{V_\manifold m}{{\cal C}_{d'}^{c_\manifold}}\right)^\frac{1}{d'}+\frac{\sigma}{\sqrt{m}}$$
The first part of our sampling condition (i.e., Eqn (\ref{dense})) is hence verified for any $\epsilon\geq\frac{1}{\sqrt{1+\frac{2}{d'}}}\left(\frac{V_\manifold m}{{\cal C}_{d'}^{c_\manifold}}\right)^\frac{1}{d'}+\frac{\sigma}{\sqrt{m}}$. 
Moreover, for any $x\in\R^d$, $\dM(x)\geq d(x,\manifold)$ because $\manifold$ is the support of $\mu_\manifold$. 
Thus:
$$d(x,\manifold)\leq\dM(x)\leq\dm(x)+\frac{\sigma}{\sqrt{m}}\leq\dm(x)+\epsilon$$
holds with no constraints on the value of $\dm(x)$. That is, for $r=\infty$, $\mu$ verifies the second part of our sampling condition (Eqn (\ref{sparse}).
This completes the proof of Theorem~\ref{tNoiseModelWas}.
\end{proof}

\section{Experimental illustration for functional noise}\label{sec:illustration}
Here, we present results obtained by applying our methods to cases where there is only functional noise.
Our goals are to demonstrate the denoising power of both the $k$-median and the disparity-based approaches and to illustrate the differences between the practical performances of the $k$-median and disparity-based denoising methods. 
We compare our denoising results with the popular k-NN algorithm, which simply sets the function at point $p$ to be the mean of the observed function values of its $k$ nearest neighbours. 
Note that, when $k' = k$, our disparity-based method is equivalent to the k-NN algorithm. 

Going back to the bone example from section~\ref{ssec:functional}, we apply our algorithm to the $10$-nearest neighbours and $k'=8$.
Using $100$ sampling of the Bone with $1000$ points each, we compute the average maximal error made by the various methods.
The disparity-based method commits a maximal error of $10\%$ on average, while the median-based method recovers the values with an error of $2\%$ and the simple $k$-NN regression gives a maximal error of $16\%$, with most error concentrated around the neck region, see Figure \ref{fig:bonedenoise}.
These results translate into the persistence diagrams that are more robust with the use of the disparity (blue squares) or the $k$-median (red diamond) instead of the $k$-NN regression (green circles), see Figure \ref{fig:bonediagram}.
Both methods retrieve the 1-dimensional topological feature.
The $k$-NN regression keeps some prominent $0$-dimensional feature through the diagram instead of having a unique component, result obtained by using the disparity or the median.
The persistence diagram of the original bone is given in red and contains only one feature.

\begin{figure}[htbp]
\begin{center}
\fbox{
\begin{tabular}{ccc}
\includegraphics[height=2.5cm]{Bone-nonoise} &
\includegraphics[height=2.5cm]{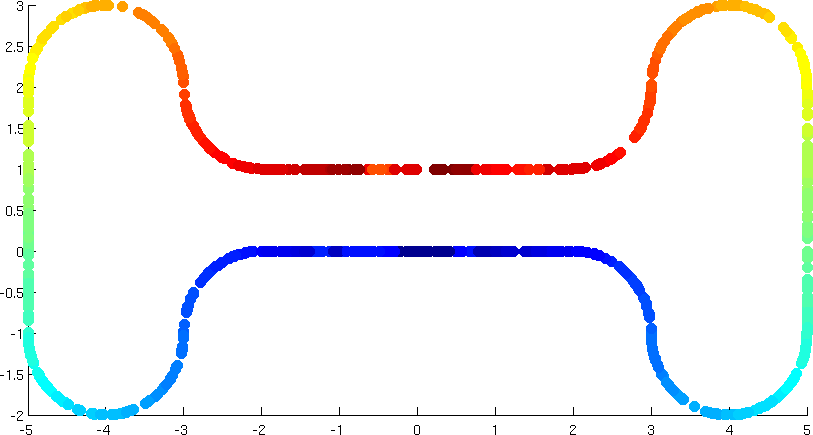} \\
 {Bone without noise} & {Bone after projection and $k$-NN} \\
\includegraphics[height=2.5cm]{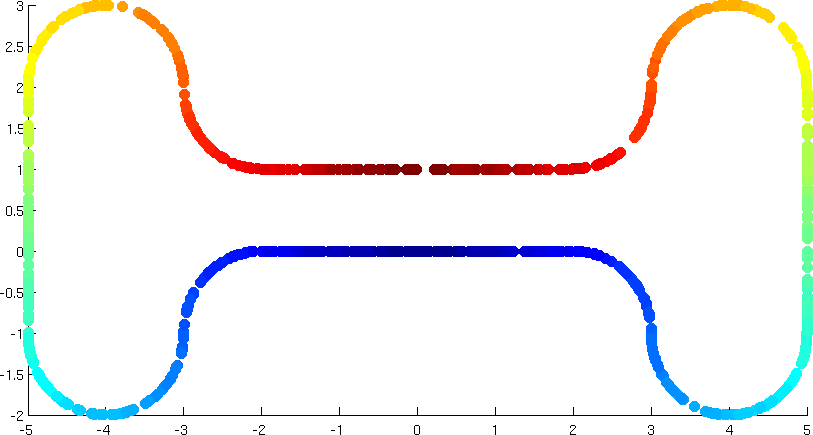} &
\includegraphics[height=2.5cm]{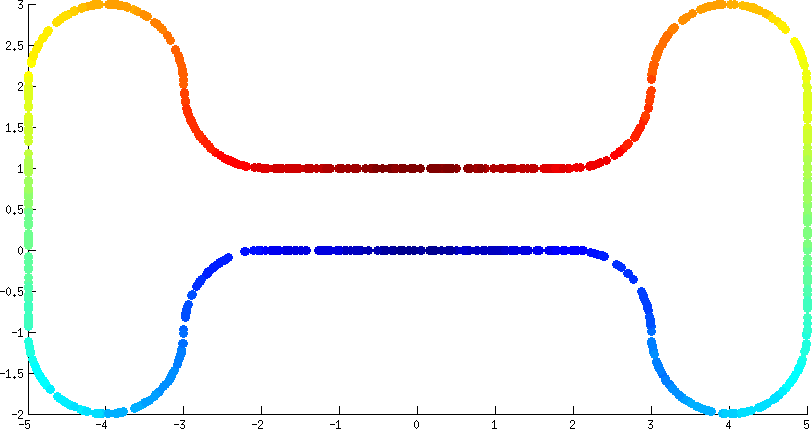}\\
{Bone after projection and disparity} & {Bone after projection and median}
\end{tabular}
}
\end{center}
\caption{Bone example after applying Gaussian perturbation, magical filter and a regression
\label{fig:bonedenoise}}
\end{figure}

\begin{figure}[htpb]
\begin{center}
\fbox{
\includegraphics[height=5cm]{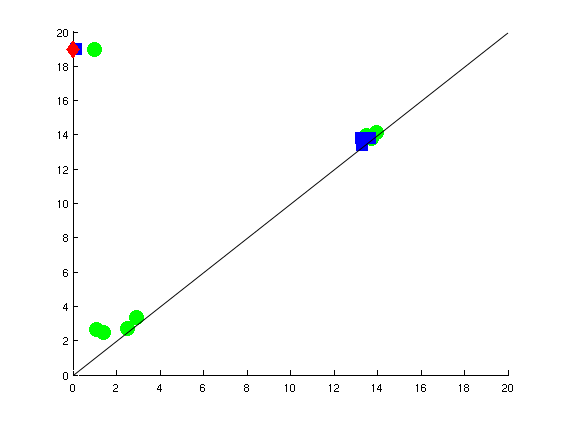}
}
\end{center}
\vspace*{-0.1in}
\caption{Persistence diagrams in dimension 0 for the Bone example: red, green and blue points constitute the $0$-th persistence diagram produced from clean (noise-less) data, from the denoised data by using $k$-NN regression, and from the denoised data by using disparity method, respectively.
\label{fig:bonediagram}}
\end{figure}

As indicated by the theoretical results, the disparity-based method improves the classic $k$-NN regression but the median-based algorithm performs slightly better.
The disparity however displays a better empirical behaviour when the Lipschitz condition on the input scalar field is relaxed, and/or the amount of noise becomes large. 
Additional illustrations can be found in the appendix.

\paragraph*{Image denoising}
We use a practical application: image denoising.
We take the greyscale image Lena as the target scalar field $f$. In Figure \ref{fig:Lena}, we use two ways to generate a noisy input scalar field $\tilde{f}$. 
The first type of noisy input is generated by adding uniform random noise as follows: with probability $p$, each pixel will receive a uniformly distributed random value in range $[0, 255]$ as its function value; otherwise, it is unchanged. 
Results under random noises are in the second and third rows of Figure \ref{fig:Lena}. 
We also consider what we call \emph{outlier noise}: with probability $p$, each pixel will be a outlier meaning that its function value is a fixed constant, which is set to be 200 in our experiments. This outlier noise is to simulate the aberrant function values caused by say, a broken sensor. The denoising results under the outlier-noise are shown in the last row of Figure \ref{fig:Lena}. 

First, we note that kNN approach tends to smooth out function values. In addition to the blurring artifact, its denoising capability is limited when the amount of noise is high (where imprecise values become dominant). 
As expected, both k-median and disparity based methods outperform the kNN approach. Indeed, they demonstrate robust recovery of the input image even with $50\%$ amount of random noise are added. 

\begin{figure}[htbp]
\begin{center}
\begin{tabular}{cc}
\includegraphics[height=3cm]{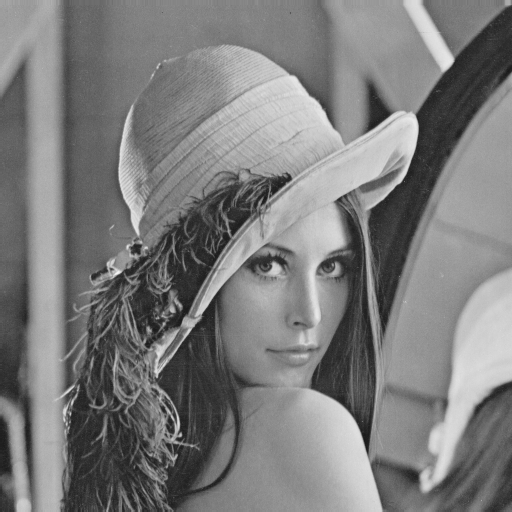} & \includegraphics[height=3cm]{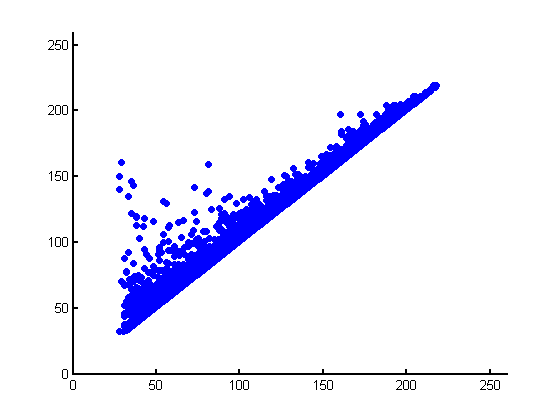} \\
{\small Original Lena} & {\small The $0$-th persistence diagram}
\end{tabular} 
\fbox{
\begin{tabular}{cccc}
\includegraphics[height=3cm]{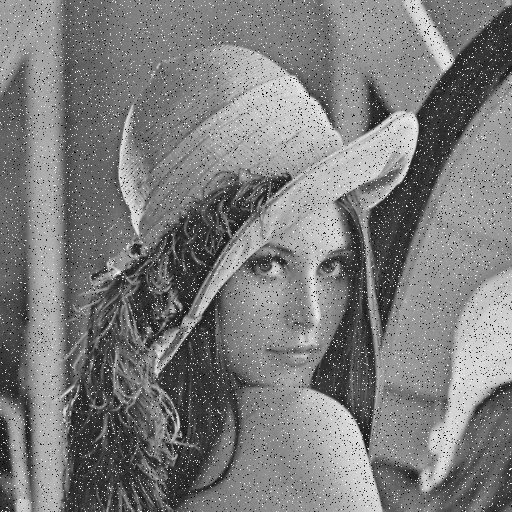} &
\includegraphics[height=3cm]{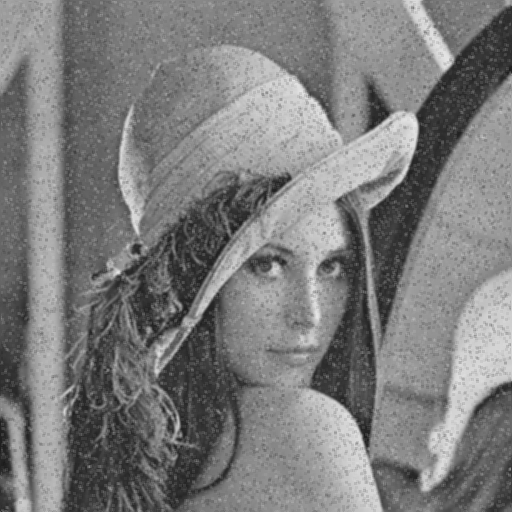} &
\includegraphics[height=3cm]{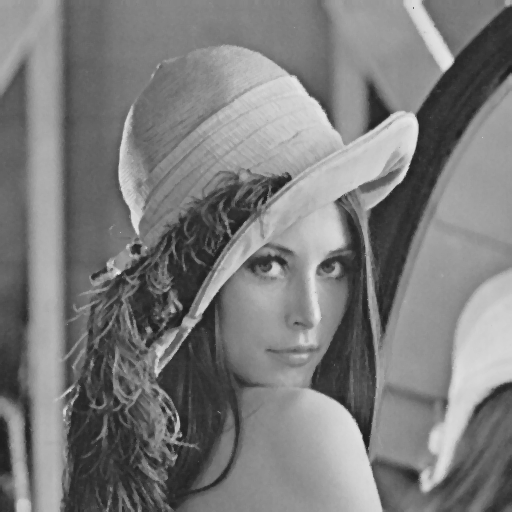} &
\includegraphics[height=3cm]{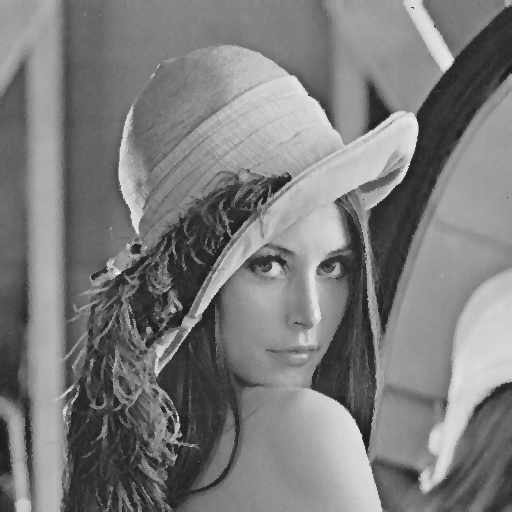} \\
{\small $10\%$ random noise} & {\small kNN: $k = 9$} & {\small k-median, $k = 9$} & {\small disparity, $k=9, k'= 5$ }\\
\includegraphics[height=3cm]{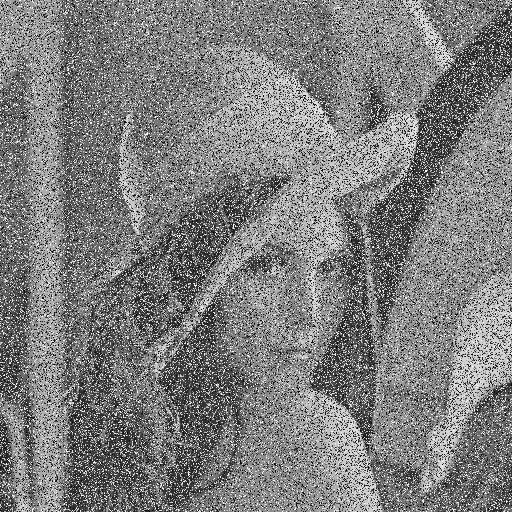} &
\includegraphics[height=3cm]{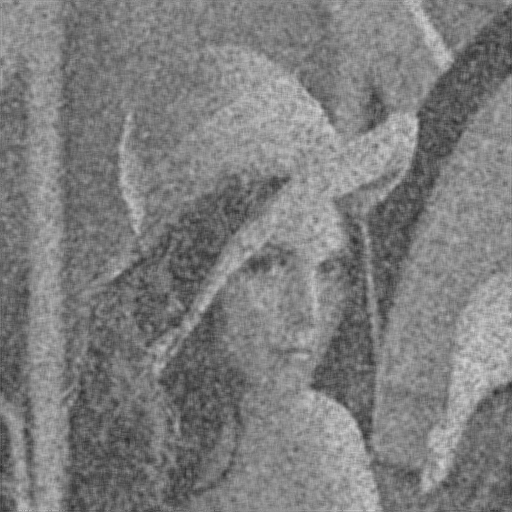} &
\includegraphics[height=3cm]{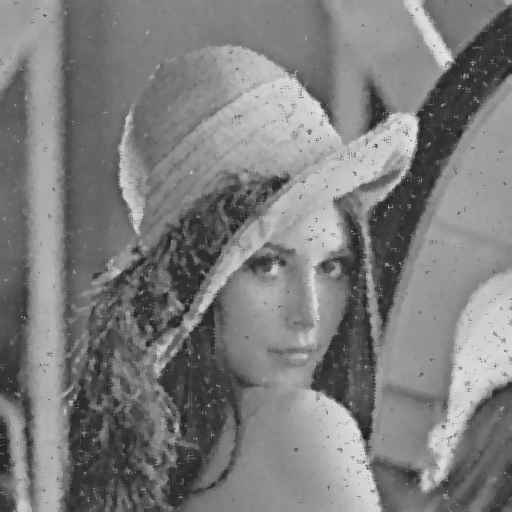} &
\includegraphics[height=3cm]{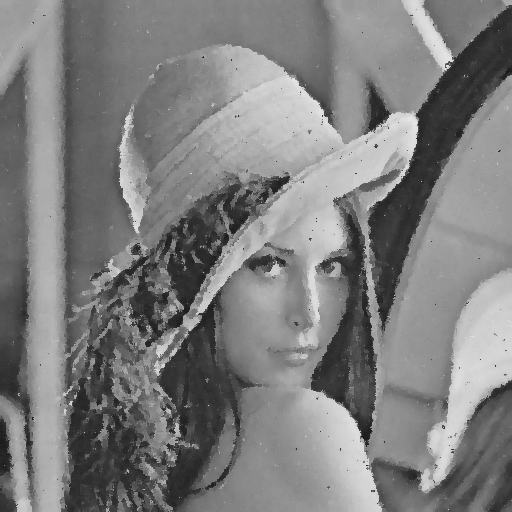} \\
 {\small $50\%$ random noise} & {\small kNN: $k = 25$} & {\small k-median, $k = 25$} & {\small disparity, $k=25, k'= 13$ }
\end{tabular}
}
\fbox{
\begin{tabular}{cccc}
\includegraphics[height=3cm]{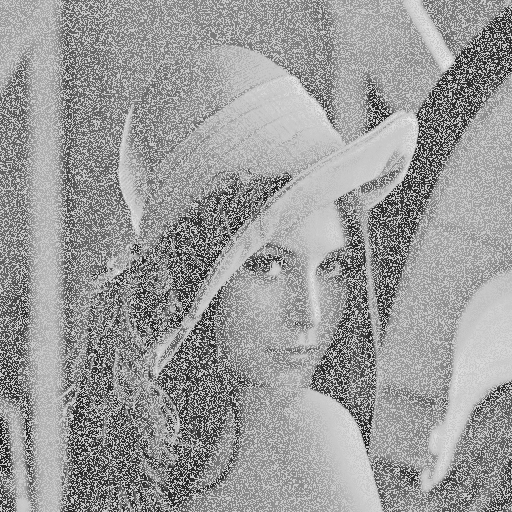} &
\includegraphics[height=3cm]{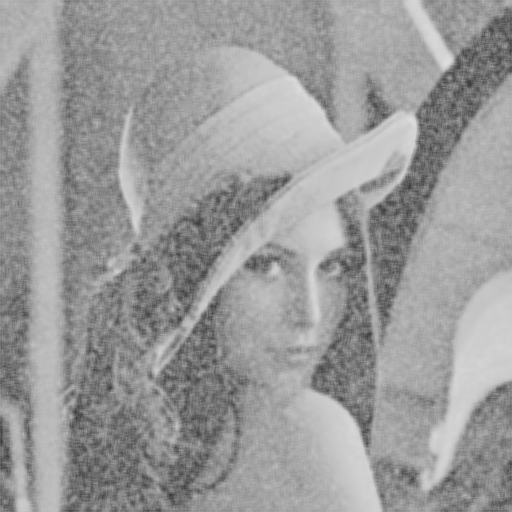} &
\includegraphics[height=3cm]{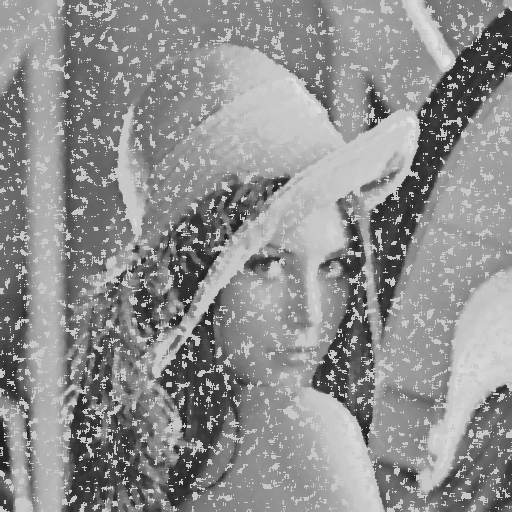} &
\includegraphics[height=3cm]{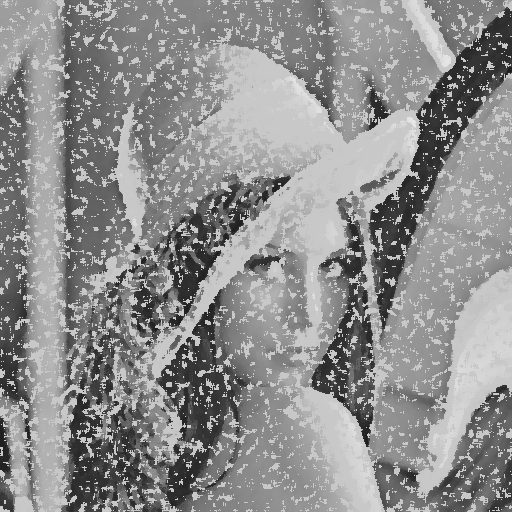} \\
\includegraphics[height=2.26cm]{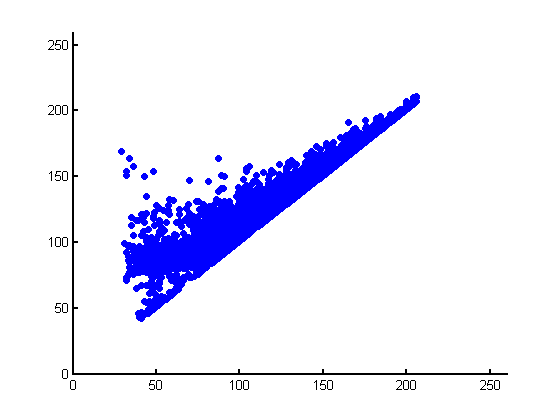} &
\includegraphics[height=2.26cm]{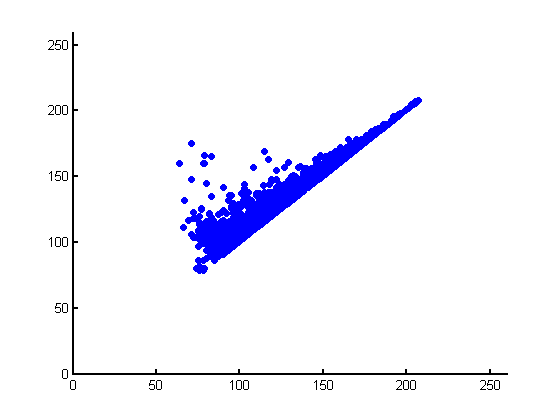} &
\includegraphics[height=2.26cm]{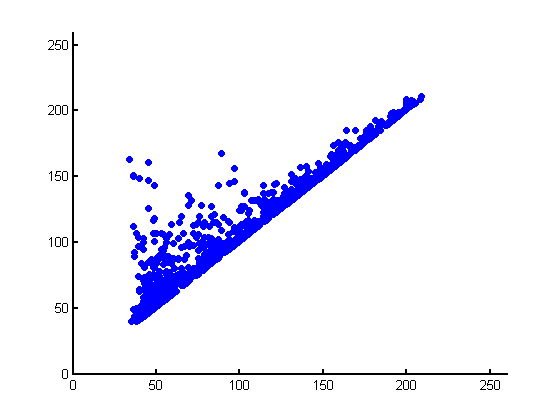} &
\includegraphics[height=2.26cm]{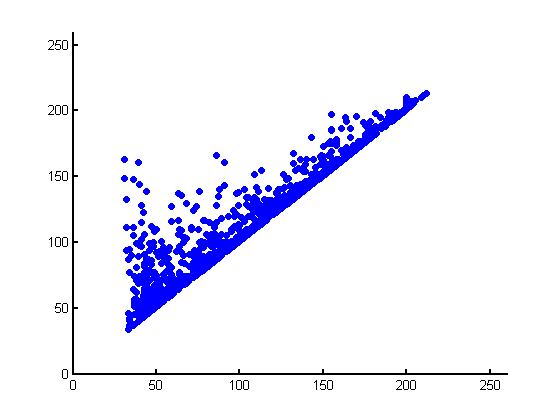} \\
 {\small $40\%$ outlier noise} & {\small kNN: $k = 25$} & {\small k-median, $k = 25$} & {\small disparity, $k=25, k'= 13$ }
\end{tabular}
}
 
\end{center}
\vspace*{-0.15in}
\caption{The denoised images after kNN, k-median, and disparity denoising approaches. 
The first row shows the original image and its $0$-th persistence diagram.
Second and third rows are under random noise of input, while fourth row are under outlier-noise as described in the text. The fifth row provides the $0$-th persistence diagrams 
on images
in the fourth row, which are 
 computed by the scalar field analysis algorithm from~\cite{sfapcdCGOS} . }
\label{fig:Lena}
\end{figure}

While both k-median and disparity based methods are more resilient against noise, there are interesting difference between their practical performances. 
From a theoretical point of view, when the input scalar field is indeed a $(k, k', \Delta)$-\funcsample{}, k-median method gives a slightly better error bound (Observation \ref{obs:kmedian-bound}) as compared to the disparity based method (Lemma \ref{lem:discrep-bound}). 
However, when $(k, k', \Delta)$-sampling condition is not satisfied, the median value can be quite arbitrary. 
By taking the average of a subset of points, the disparity method, on the other hand, is more robust against large amount of noise. 
This difference is evident in the third and last row of Figure \ref{fig:Lena}. 

%

Moreover, the application to persistent homology which was our primary goal is much cleaner after the disparity-based method.
The structure of the beginning of the diagrams is almost perfectly retrieved by both the median and disparity-based methods. 
However, the median induces a shrinking phenomenon to the diagram. 
This means that the width of the diagram is reduced ans so are the lifespans of topological features, making it more difficult to distinguish between noise and relevant information.
We remark that the classic $k$-NN approach shrinks the diagram even more, to the point that it is very hard to distinguish the information from the noise.

The standard indicator to measure the quality of a denoising is the \emph{Peak Signal over Noise Ratio} (PSNR).
Given a grey scale input image $I$ and an output image $O$ with the grey scale between $0$ and $255$, it is defined by
$$\mathrm{PSNR}(I,O)=10\log_{10}\left(\frac{256^2}{\frac{1}{ij}\sum_{i}\sum_{j}(I[i][j]-O[i][j])^2}\right).$$
Figure~\ref{fig:stats_Lena} shows the quality of the denoising for a set of Lena images with increasing quantity of noise.
The curves are obained using the median ($M$) and different values of $k'$ in the disparity while $k$ is fixed at $25$.
The median is better when the noise ratio is small but as we increase the number of outliers, the disparity obtains better results.
This also shows that the optimal $k'$ depends on the noise ratio.
It also depends on the image we consider and thus makes it difficult to find an easy way to choose it automatically.
Heuristically, it is better to take $k'$ around $\frac{2}{3}k$, especially when there is a lot of noise.

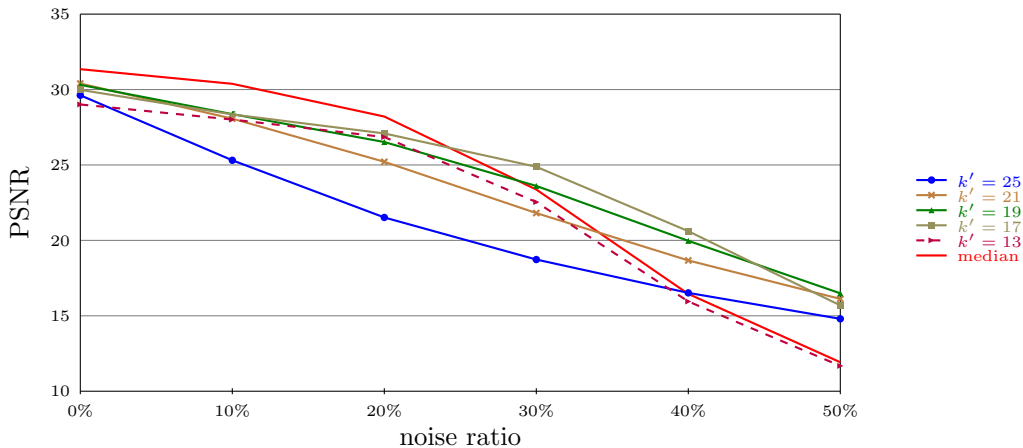
\begin{figure}[ht]
\begin{center}
\begin{tikzpicture}[scale=.2]
\draw[thin] (0,0) node[left] {\tiny$10$} -- (50,0);
\draw[thin] (0,5) node[left] {\tiny$15$};
\draw[thin] (0,10) node[left] {\tiny$20$};
\draw[thin] (0,15) node[left] {\tiny$25$} ;
\draw[thin] (0,20) node[left] {\tiny$30$};
\draw[thin] (0,25) node[left] {\tiny$35$} -- (50,25);
\draw[thin,black!50] (0,5) -- (50,5);
\draw[thin,black!50] (0,10) -- (50,10);
\draw[thin,black!50] (0,15) -- (50,15);
\draw[thin,black!50] (0,20) -- (50,20);
\draw[thin] (0,-.2) node[below] {\tiny$0\%$} -- (0,25);
\draw[thin] (10,-.2) node[below] {\tiny$10\%$} -- (10,.2);
\draw[thin] (20,-.2) node[below] {\tiny$20\%$} -- (20,.2);
\draw[thin] (30,-.2) node[below] {\tiny$30\%$} -- (30,.2);
\draw[thin] (40,-.2) node[below] {\tiny$40\%$} -- (40,.2);
\draw[thin] (50,-.2) node[below] {\tiny$50\%$} -- (50,25);

\draw[thick,red] (0,21.35) -- (10,20.38) -- (20,18.21) -- (30,13.37) -- (40,6.45) -- (50,1.92);

\draw[thick,blue] (0,19.62) -- (10,15.31) -- (20,11.52) -- (30,8.73) -- (40,6.52) -- (50,4.8);
\fill[blue] (0,19.62) circle (7pt);
\fill[blue] (10,15.31) circle (7pt);
\fill[blue] (20,11.52) circle (7pt);
\fill[blue] (30,8.73)  circle (7pt);
\fill[blue] (40,6.52) circle (7pt);
\fill[blue] (50,4.8) circle (7pt);

\draw[thick,brown] (0,20.4) -- (10,18.07) -- (20,15.21) -- (30,11.81) -- (40,8.67) -- (50,6.12);
\draw[thick,brown] (-.2,20.2) -- (.2,20.6);
\draw[thick,brown] (.2,20.2) -- (-.2,20.6);
\draw[thick,brown] (9.8,17.87) -- (10.2,18.27);
\draw[thick,brown] (10.2,17.87) -- (9.8,18.27);
\draw[thick,brown] (19.8,15.01) -- (20.2,15.41);
\draw[thick,brown] (20.2,15.01) -- (19.8,15.41);
\draw[thick,brown] (29.8,11.61) -- (30.2,12.01);
\draw[thick,brown] (30.2,11.61) -- (29.8,12.01);
\draw[thick,brown] (39.8,8.47) -- (40.2,8.87);
\draw[thick,brown] (40.2,8.47) -- (39.8,8.87);
\draw[thick,brown] (49.8,5.92) -- (50.2,6.32);
\draw[thick,brown] (50.2,5.92) -- (49.8,6.32);

\draw[thick,green!50!black!100] (0,20.31) -- (10,18.37) -- (20,16.52) -- (30,13.61) -- (40,9.97) -- (50,6.49);
\fill[green!50!black!100] (-.2,20.01) -- (0,20.51) -- (.2,20.01);
\fill[green!50!black!100] (9.8,18.17) -- (10,18.57) -- (10.2,18.57);
\fill[green!50!black!100] (19.8,16.32) -- (20,16.72) -- (20.2,16.32);
\fill[green!50!black!100] (29.8,13.41) -- (30,13.81) -- (30.2,13.41);
\fill[green!50!black!100] (39.8,9.77) -- (40,10.17) -- (40.2,9.77);
\fill[green!50!black!100] (49.8,6.29) -- (50,6.69) -- (50.2,6.29);

\draw[thick,yellow!50!black!100] (0,19.98) -- (10,18.35) -- (20,17.09) -- (30,14.88) -- (40,10.61) -- (50,5.69);
\fill[yellow!50!black!100] (-.2,19.78) -- (.2,19.78) -- (.2,20.18) -- (-.2,20.18);
\fill[yellow!50!black!100] (9.8,18.15) -- (10.2,18.15) -- (10.2,18.55) -- (9.8,18.15);
\fill[yellow!50!black!100] (19.8,16.89) -- (20.2,16.89) -- (20.2,17.29) -- (19.8,17.29);
\fill[yellow!50!black!100] (29.8,14.68) -- (30.2,14.68) -- (30.2,15.08) -- (29.8,15.08);
\fill[yellow!50!black!100] (39.8,10.41) -- (40.2,10.41) -- (40.2,10.81) -- (39.8,10.81);
\fill[yellow!50!black!100] (49.8,5.49) -- (50.2,5.49) -- (50.2,5.89) -- (49.8,5.89);

\draw[thick,purple,dashed] (0,19.01) -- (10,18.01) -- (20,16.85) -- (30,12.54) -- (40,5.95) -- (50,1.68);
\fill[purple] (-.2,19.21) -- (-.2,18.81) -- (.2,19.01);
\fill[purple] (9.8,18.21) -- (9.8,17.81) -- (10.2,18.01);
\fill[purple]  (19.8,17.05) -- (19.8,16.65) -- (20.2,16.85);
\fill[purple]  (29.8,12.74) -- (29.8,12.34) -- (30.2,12.54);
\fill[purple]  (39.8,6.15) -- (39.8,5.75) -- (40.2,5.95);
\fill[purple]  (49.8,1.88) -- (49.8,1.48) -- (50.2,1.68);

\draw (25,-3) node {\small noise ratio};
\draw (-4,12.5) node[rotate=90]  {\small PSNR};

\draw[thick,red] (55,9) -- (57,9) node[right] {\tiny median};
\draw[thick,blue] (55,14) -- (57,14) node[right] {\tiny $k'=25$};
\fill[blue] (56,14) circle (7pt);
\draw[thick,brown] (55,13) -- (57,13) node[right] {\tiny $k'=21$};
\draw[thick,brown] (55.8,12.8) -- (56.2,13.2);
\draw[thick,brown] (56.2,12.8) -- (55.8,13.2);
\draw[thick,green!50!black!100] (55,12) -- (57,12) node[right] {\tiny $k'=19$};
\fill[green!50!black!100] (55.8,11.8) -- (56,12.2) -- (56.2,11.8);
\draw[thick,yellow!50!black!100] (55,11) -- (57,11) node[right] {\tiny $k'=17$};
\fill[yellow!50!black!100] (55.8,10.8) -- (56.2,10.8) -- (56.2,11.2) -- (55.8,11.2);
\draw[thick,purple,dashed] (55,10) -- (57,10) node[right] {\tiny $k'=13$};
\fill[purple] (55.8,10.2) -- (55.8,9.8) -- (56.2,10);
\end{tikzpicture}
\caption{PSNR for Lena images depending on the choice of $k'$ and the quantity of noise}\label{fig:stats_Lena}
\end{center}
\end{figure}

State of the art results in computer vision obtain better experimental results (e.g.~\cite{ndwmfrrvinDX,dspncimdwmfCL,nidfmrwrinWW}).
However, these results assume that the noise model is known and they can start by detecting and removing noisy points before rebuilding the image.
Our methods are free from assumptions on the generative model of the image.
The algorithms do not change depending on the type of noise.

\paragraph*{Persistence diagram computation}
We consider a more topological example from real data.
We consider an elevation map of an area near Corte in the French island of Corsica.
The true measures of elevation are given in the left image of Figure~\ref{fig:Corte}.
The topography can be analysed by looking at the function minus-altitude.
We add random faulty sensors that give false results with a $20\%$ probability to simulate malfunctioning equipments.
The area covers a square of 2 minutes of arc in both latitude and longitude.
We apply our algorithm with the following parameters: $k=9$, $k'=7$, $\eta=.05$ minute and $\delta=.025$ minute.
We show the recovered persistence diagrams in Figure~\ref{fig:CorteDiagram}, where the prominent peaks of the original elevation map are highlighted. 
The ``gap'' stands for the ratio between the shortest living relevant feature, highlighted in red, and the longest feature created by the noise.

\begin{figure}[ht]
\begin{center}
\fbox{
\begin{tabular}{ccc}
\includegraphics[height=3cm]{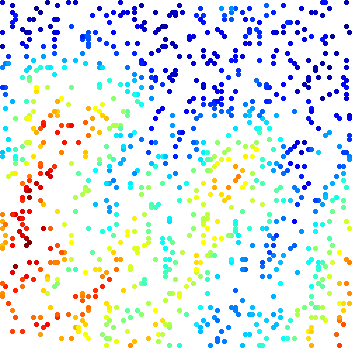} & \includegraphics[height=3cm]{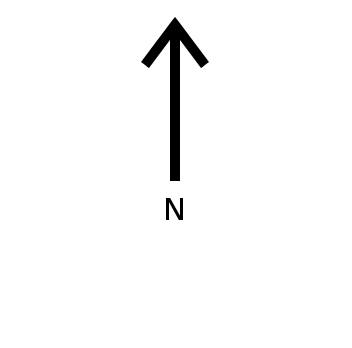} & \includegraphics[height=3cm]{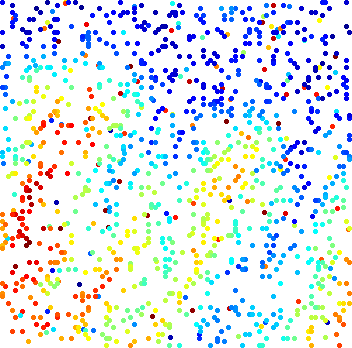}\\
\tiny{Without noise} & & \tiny{With 20\% background noise}
\end{tabular}}
\end{center}
\vspace{-.2in}
\caption{Elevation map around Corte}
\label{fig:Corte}
\end{figure}
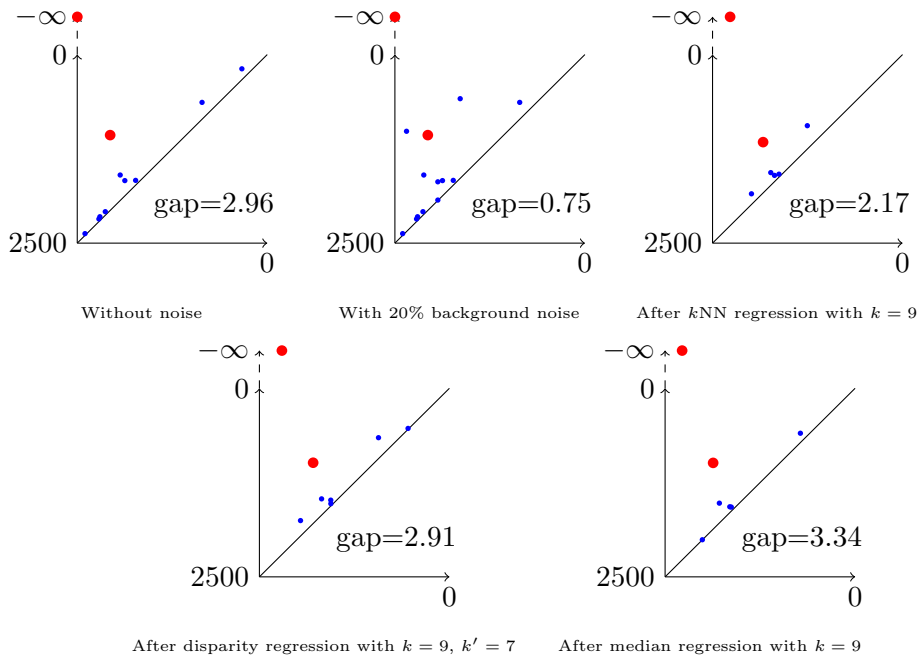
\begin{figure}[ht]
\begin{center}

\begin{tabular}{ccc}
\begin{tikzpicture}[scale=.1]
\draw[thin,->] (0,0) -- (25,0);
\draw[thin,->] (0,0) -- (0,25);
\draw[thin] (0,0) -- (25,25);
\draw[thin,dashed,->] (0,25) -- (0,30);
\draw (0,0) node[below,left] {$2500$};
\draw (25,0) node[below] {$0$};
\draw (0,25) node[left] {$0$};
\draw (0,30) node[left] {$-\infty$};
\fill[red] (.01,30) circle (20pt);
\fill[red] (4.33,14.32) circle (20pt);
\fill[blue] (21.66,23.11) circle (10pt);
\fill[blue] (16.42,18.66) circle (10pt);
\fill[blue] (7.68,8.32) circle (10pt);
\fill[blue] (6.26,8.29) circle (10pt);
\fill[blue] (5.65,9.03) circle (10pt);
\fill[blue] (3.69,4.17) circle (10pt);
\fill[blue] (2.98,3.5) circle (10pt);
\fill[blue] (2.85,3.2) circle (10pt);
\fill[blue] (1.02,1.26) circle (10pt);
\draw (18,5) node {gap=$2.96$};
\end{tikzpicture}
&
\begin{tikzpicture}[scale=.1]
\draw[thin,->] (0,0) -- (25,0);
\draw[thin,->] (0,0) -- (0,25);
\draw[thin] (0,0) -- (25,25);
\draw[thin,dashed,->] (0,25) -- (0,30);
\draw (0,0) node[below,left] {$2500$};
\draw (25,0) node[below] {$0$};
\draw (0,25) node[left] {$0$};
\draw (0,30) node[left] {$-\infty$};
\fill[red] (.01,30) circle (20pt);
\fill[red] (4.33,14.32) circle (20pt);
\fill[blue] (16.42,18.66) circle (10pt);
\fill[blue] (8.59,19.15) circle (10pt);
\fill[blue] (7.68,8.32) circle (10pt);
\fill[blue] (6.26,8.29) circle (10pt);
\fill[blue] (5.65,8.13) circle (10pt);
\fill[blue] (5.65,5.71) circle (10pt);
\fill[blue] (3.80,9.03) circle (10pt);
\fill[blue] (3.69,4.17) circle (10pt);
\fill[blue] (2.98,3.5) circle (10pt);
\fill[blue] (2.85,3.2) circle (10pt);
\fill[blue] (1.53,14.83) circle (10pt);
\fill[blue] (1.02,1.26) circle (10pt);
\draw (18,5) node {gap=$0.75$};
\end{tikzpicture}
&
\begin{tikzpicture}[scale=.1]
\draw[thin,->] (0,0) -- (25,0);
\draw[thin,->] (0,0) -- (0,25);
\draw[thin] (0,0) -- (25,25);
\draw[thin,dashed,->] (0,25) -- (0,30);
\draw (0,0) node[below,left] {$2500$};
\draw (25,0) node[below] {$0$};
\draw (0,25) node[left] {$0$};
\draw (0,30) node[left] {$-\infty$};
\fill[red] (2.31,30) circle (20pt);
\fill[red] (6.65,13.39) circle (20pt);
\fill[blue] (12.47,15.57) circle (10pt);
\fill[blue] (8.74,9.14) circle (10pt);
\fill[blue] (8.15,8.98) circle (10pt);
\fill[blue] (7.66,9.35) circle (10pt);
\fill[blue] (5.12,6.54) circle (10pt);
\draw (18,5) node {gap=$2.17$};
\end{tikzpicture}\\
\tiny{Without noise} & \tiny{With 20\% background noise} &
\tiny{After $k$NN regression with $k=9$}
\end{tabular}
\newline
\begin{tabular}{cc}
\begin{tikzpicture}[scale=.1]
\draw[thin,->] (0,0) -- (25,0);
\draw[thin,->] (0,0) -- (0,25);
\draw[thin] (0,0) -- (25,25);
\draw[thin,dashed,->] (0,25) -- (0,30);
\draw (0,0) node[below,left] {$2500$};
\draw (25,0) node[below] {$0$};
\draw (0,25) node[left] {$0$};
\draw (0,30) node[left] {$-\infty$};
\fill[red] (2.96,30) circle (20pt);
\fill[red] (7.07,15.13) circle (20pt);
\fill[blue] (19.56,19.68) circle (10pt);
\fill[blue] (15.67,18.45) circle (10pt);
\fill[blue] (9.40,9.70) circle (10pt);
\fill[blue] (9.39,10.17) circle (10pt);
\fill[blue] (8.18,10.34) circle (10pt);
\fill[blue] (5.41,7.46) circle (10pt);
\draw (18,5) node {gap=$2.91$};
\end{tikzpicture}\hspace{1em}&\hspace{1em}
\begin{tikzpicture}[scale=.1]
\draw[thin,->] (0,0) -- (25,0);
\draw[thin,->] (0,0) -- (0,25);
\draw[thin] (0,0) -- (25,25);
\draw[thin,dashed,->] (0,25) -- (0,30);
\draw (0,0) node[below,left] {$2500$};
\draw (25,0) node[below] {$0$};
\draw (0,25) node[left] {$0$};
\draw (0,30) node[left] {$-\infty$};
\fill[red] (2.24,30) circle (20pt);
\fill[red] (6.32,15.11) circle (20pt);
\fill[blue] (17.8,19.04) circle (10pt);
\fill[blue] (8.71,9.25) circle (10pt);
\fill[blue] (8.51,9.29) circle (10pt);
\fill[blue] (7.14,9.77) circle (10pt);
\fill[blue] (4.9,4.92) circle (10pt);
\draw (18,5) node {gap=$3.34$};
\end{tikzpicture}\\
\tiny{After disparity regression with $k=9$, $k'=7$} & \tiny{After median regression with $k=9$}
\end{tabular}

\end{center}
\vspace{-.2in}
\caption{Persistence diagrams of Corte Elevation map}
\label{fig:CorteDiagram}
\end{figure}

We note that the gap in the case of the noisy point cloud (before denoising) is less than $1$.
This means that some relevant topological feature has a shorter lifespan than one caused by noise.
Intuitively, this means that it is difficulty to tell true features from noise from this persistence diagram, without performing denoising. 
We also show the persistence diagrams, as well as the ``gap" values, for the denoised data after the three denoising method: $k$-NN regression, $k$-median and our disparity based method. 
In the case of the $k$-NN regression, the topological feature are in the right order. 
However, the prominence given by the gap is significantly smaller than the one from the original point cloud.
Both the disparity based method and the median provides gaps on par with the non-noisy input and thus allow a good recovery of the correct topology.

\end{document}